\documentclass{article}

\usepackage{arxiv}

\usepackage[utf8]{inputenc} 
\usepackage[T1]{fontenc}    
\usepackage{lmodern}

\usepackage{hyperref}       
\usepackage{doi}

\usepackage{setspace}

\usepackage{enumitem}
\usepackage{natbib}
\bibpunct[:]{(}{)}{;}{a}{,}{,}

\usepackage{graphicx,color}

\usepackage{url}
\usepackage{booktabs}
\usepackage{tabularx}

\usepackage{lscape}
\usepackage{wrapfig}
\usepackage{tikz}
\usepackage{float}
\usepackage{multirow}
\usepackage{amsmath, amssymb, amsfonts, amsthm}
\usepackage{physics}
\usepackage{bm}
\usepackage{dsfont}

\newcommand{\argmax}{\mathop{\rm arg~max}\limits}
\newcommand{\argmin}{\mathop{\rm arg~min}\limits}
\newcommand{\indep}{\perp\!\!\!\perp}

\theoremstyle{plain}
\newtheorem{theorem}{Theorem}
\newtheorem*{theorem*}{Theorem}
\newtheorem{definition}{Definition}
\newtheorem*{definition*}{Definition}
\newtheorem{lemma}{Lemma}
\newtheorem*{lemma*}{Lemma}

\newtheorem*{corollary*}{Corollary}
\newtheorem{assumption}{Assumption}
\newtheorem*{assumption*}{Assumption}
\newtheorem{proposition}{Proposition}
\newtheorem*{proposition*}{Proposition}

\numberwithin{equation}{section}

\usepackage[noend]{algpseudocode}
\usepackage{algorithm}


\usepackage[subrefformat=parens]{subcaption}
\usepackage[subrefformat=parens]{caption}

\title{Identification and Estimation under Multiple Versions of Treatment: Mixture-of-Experts Approach}

\author{
  \href{https://orcid.org/0009-0001-6856-481X}
  {\includegraphics[scale=0.06]{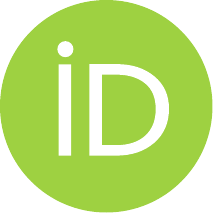}\hspace{1mm}Kohei Yoshikawa}
  \thanks{Graduate School of Mathematics, Kyushu University, 744, Motooka, Nishi-ku, Fukuoka, 819-0395 Japan. e-mail: (\texttt{yoshikawa.kohei.603 (at) s.kyushu-u.ac.jp})}
  \thanks{NTT DATA Mathematical Systems Inc., 1F Shinanomachi Rengakan, 35 Shinanomachi, Shinjuku-ku, Tokyo 160-0016, Japan. e-mail: (\texttt{yoshikawa (at) msi.co.jp})}
	\And
	\href{https://orcid.org/0000-0002-0804-0141}
  {\includegraphics[scale=0.06]{orcid.pdf}\hspace{1mm}Shuichi Kawano}
	\thanks{Faculty of Mathematics, Kyushu University, 744 Motooka, Nishi-ku, Fukuoka, 819-0395 Japan. e-mail: (\texttt{skawano (at) math.kyushu-u.ac.jp})}
}

\hypersetup{
pdftitle={Identification and Estimation under Multiple Versions of Treatment: Mixture-of-Experts Approach},
pdfauthor={Kohei Yoshikawa, Shuichi Kawano},
pdfkeywords={causal inference, multiple versions of treatment, compound treatments, mixture-of-experts, EM algorithm},
}

\begin{document}
\maketitle
%
%

\begin{abstract}
The Stable Unit Treatment Value Assumption (SUTVA) includes the condition that there are no multiple versions of treatment in causal inference.
Though we could not control the implementation of treatment in observational studies, multiple versions may exist in the treatment.
It has been pointed out that ignoring such multiple versions of treatment can lead to biased estimates of causal effects, but a causal inference framework that explicitly deals with the unbiased identification and estimation has not been fully developed yet.
Thus, it is difficult to obtain a deeper understanding for mechanisms of the complex treatments.
In this paper, we introduce the Mixture-of-Experts framework into causal inference to estimate causal contrasts between underlying versions of a treatment, even when the versions are not observed.
Numerical experiments demonstrate the effectiveness of the proposed method.
\end{abstract}

%
%

%
%
\section{Introduction}
\label{sec:introduction}

In the theory of causal inference, a fundamental starting point is the potential outcomes framework since \citet{Rubin1980-hl}, whose core assumption is the Stable Unit Treatment Value Assumption (SUTVA).
SUTVA typically comprises two components: (i) no interference, meaning that one unit's treatment does not affect another unit's outcomes; and (ii) no multiple versions of treatment, under which each treatment is uniquely defined with no variation in its implementation. 
The latter assumption is often implicitly employed to simplify the theory, but it frequently fails in applied research.
For example, in education and medical interventions, as well as in social policy, the same nominal treatment can be delivered through heterogeneous modes or patterns \citep{VanderWeele2009-xa}.
Ignoring such treatment versions when estimating causal effects may lead to ambiguous results and undermine both interpretability and external validity \citep{VanderWeele2013-lw}.

To address this issue, a framework that treats treatments with multiple versions as compound treatments has been proposed \citep{Hernan2011-cy}.
Under conventional approaches, conditioning on covariates associated with the treatment and its versions identifies causal contrasts between compound treatments, which correspond to pairwise average causal effects \citep{Feng2012-zc, McCaffrey2013-ds}.
In this case, the estimand carries a coherent interpretation as a version-mixed pairwise average causal effect and may be sufficient when researchers have no substantive interest in versions \citep{VanderWeele2013-lw}.
However, when versions substantively characterize the content or the mode of implementation of the treatment, a simple mixture average can mask heterogeneity across versions, imposing serious limits on the interpretability.
Policymakers and practitioners seek not only the overall treatment-level comparisons but also answers to questions such as which implementation is particularly effective and whether some versions may be harmful.
Addressing these questions requires methods for identifying and estimating pairwise contrasts between the underlying versions of treatment, but such methods remain underdeveloped.

When these versions are unobserved, estimating their pairwise contrasts motivates a statistical representation of the underlying version structure.
We represent this structure using a finite mixture model for the outcome distribution.
However, the statistical recovery of such a mixture structure does not immediately imply that the components of the mixture model correspond to substantively meaningful versions of treatment \citep{VanderWeele2022-cm}.

To clarify this difference, we distinguish between the substantive version and the statistical version.
A substantive version refers to a substantively meaningful treatment implementation that can, at least in principle, be represented as a recordable implementation.
Examples include a specific educational curriculum, dosing schedule, or treatment protocol.
A statistical version, in contrast, refers to a latent component that captures systematic heterogeneity of treatment implied by the data-generating process.
To interpret statistical versions in terms of substantive versions, additional conditions linking the two are required.
Under the causal-alignment condition formalized in Assumption \ref{asmp:causal-alignment}, we refer to a statistical version as a ``latent version'', for which a causal interpretation in terms of substantive versions is justified.

The use of latent components to represent treatment versions connects our framework to the literature on latent treatments.
\citet{Kang2010-kl, Schuler2014-zo} developed approaches for estimating a latent treatment variable through latent class analysis.
These approaches use external indicators to identify latent treatment classes and estimate their causal effects.
This line of work has also led to subsequent applications of causal inference with latent treatments \citep{Kang2020-xx, Suk2021-od, Le2025-iu}.
Related work has considered statistically learned representations of complex interventions.
For example, \citet{Fong2023-qs} maps high-dimensional interventions to lower-dimensional latent treatments and imposes conditions connecting the learned treatment representation to outcome responses.

Our setting differs from these latent-treatment formulations.
Here, the nominal treatment is directly observed, whereas heterogeneity in its implementation is unobserved and appears as versions nested within each observed treatment.
This nested structure induces a mixture distribution within each treatment group.
To interpret the components of this mixture distribution in terms of substantive versions, we impose a structural condition linking substantive implementations to the latent variable of the mixture model.

Another framework in which mixture models have been used is principal stratification, which defines causal effects within subpopulations characterized by the post-treatment variable \citep{Frangakis2002-ps}.
Identification and estimation within principal strata have been studied using principal ignorability \citep{Stuart2015-pi,Jiang2021-ps} and parametric Gaussian mixture models \citep{Jo2022-gm}.
In our framework, the latent components represent substantive treatment versions under the causal-alignment condition, whereas principal strata represent subgroups of individuals.

In this study, we construct the theoretical framework for the mixture modeling with multiple versions of treatment problem in causal inference.
Specifically, we model the multiple versions that may exist within each treatment as latent variables and propose a framework that employs mixture-of-experts (MoE) \citep{Jacobs1991-rx} for identification and estimation.
In this paper, we focus on identifying and estimating pairwise average contrasts between such latent versions, which we refer to as the pairwise average causal effect (ACE).
Our approach offers two primary contributions.
First, the framework identifies the pairwise ACE under an alignment condition between substantive implementations and statistical components of MoE.
Second, we develop a maximum-likelihood estimation procedure via the expectation-maximization (EM) algorithm and asymptotic theory of the estimator for pairwise ACE under standard regularity and identifiability conditions.
In conclusion, the proposal aims to achieve both accurate effect estimation and interpretability in observational data where the existence of multiple versions of treatment is practically unavoidable.

The remainder of the paper is organized as follows.
Section \ref{sec:preliminaries} introduces the model setup and identification conditions.
Section \ref{sec:proposed-method} introduces the MoE representation, causal-alignment condition, and estimation procedure.
Section \ref{sec:simulation-study} presents simulation studies, while Section \ref{sec:empirical-analysis} presents an empirical study. 
Section \ref{sec:discussions-and-extensions} concludes with a discussion and directions for future research.
Technical proofs are provided in the appendices.

%
%
\section{Preliminaries}
\label{sec:preliminaries}

Suppose we observe $n$ independent and identically distributed units.
For each unit $i$, we consider a treatment $T_i$ taking values in the finite set $\mathcal{T}=\{0,1,\dots,J-1\}$, where $J \ge 2$ denotes the number of treatments.
Each treatment $t \in \mathcal{T}$ has multiple versions. 
The version received by unit $i$ under the condition $T_i=t$ is denoted by $K_i^{(t)} \in \mathcal{K}^{(t)}=\{0,1,\ldots,J^{(t)}-1\}$, where $J^{(t)} \ge 2$ denotes the number of versions associated with treatment $t$.

Here, we define the treatment indicator function $D_i^{(t)}$ as:
\begin{align*}
    D_{i}^{(t)} = \begin{cases}
        1, & \text{if}\ T_i = t, \\
        0, & \text{otherwise}.
    \end{cases}
\end{align*}
We define the joint treatment-version indicator $Z_i^{(t,k)}$ as:
\begin{align*}
    Z_{i}^{(t, k)} = \begin{cases}
        1, & \text{if}\ T_i=t \text{ and } K_i^{(t)}=k,\\
        0, & \text{otherwise}.
    \end{cases}
\end{align*}
These indicator functions represent the exclusivity of the treatment.
Specifically, each unit is conceptually assigned to exactly one treatment-version pair.

By using the potential outcome framework \citep{Rubin1974-zj}, we define a set of potential outcomes as $\{Y_i^{(t, k)} \mid t \in \mathcal{T}, k \in \mathcal{K}^{(t)}\}$, where $Y_i^{(t,k)}$ is a random variable that maps each treatment $t \in \mathcal{T}$ and version $k \in \mathcal{K}^{(t)}$ to a potential outcome for each unit $i$.
The observed outcome associated with the only potential outcome is denoted by $Y_i \equiv Y_i^{(t, k)}$.
This notation implicitly assumes the SUTVA, described by \citet{Rubin1978-ht}.
The SUTVA comprises no interference across units and well-defined treatments and versions.
The assumptions ensure the corresponding potential outcome $Y_i^{(t, k)}$ is uniquely determined when $T_i=t$ and $K_i^{(t)}=k$.
We further introduce covariate vectors $\bm{X}_i \in \mathbb{R}^p$ for each unit $i$.
The covariates $\bm{X}_i$ are common confounders that may affect $T_i$, $K_i^{(t)}$, and $Y_i^{(t, k)}$ for all $t \in \mathcal{T}, k \in \mathcal{K}^{(t)}$.
Throughout, we assume that $\mathbb E\qty[|Y_i^{(t,k)}|]<\infty$ for every $t\in\mathcal T$ and $k\in\mathcal K^{(t)}$.

Under the following standard assumptions, the mean potential outcome $E\qty[Y_i^{(t,k)}]$ for treatment-version pair $(t,k)$ can be identified.
Each assumption represents a generalization of widely used concepts in observational studies as settings involving multiple versions of treatment \citep{Hernan2011-cy, VanderWeele2013-lw}.
\begin{assumption}[Consistency]
  For all units $i$,
  $$
  Y_i = Y_i^{(t,k)} \quad \text{ when } T_i=t \text{ and } K_i^{(t)}=k.
  $$
  \label{asmp:consistency}
\end{assumption}
\begin{assumption}[Weak conditional exchangeability]
  For each $t \in \mathcal{T}$ and $k \in \mathcal{K}^{(t)}$,
  $$
  Y_i^{(t,k)} \indep Z_i^{(t,k)} \mid \bm{X}_i.
  $$
  \label{asmp:weak-conditional-exchangeability}
\end{assumption}
Assumption \ref{asmp:consistency} links the observed outcome $Y_i$ to the potential outcomes $Y_i^{(t,k)}$ by requiring that the treatment $(t,k)$ be well-defined.
Specifically, when $T_i=t$ and $K_i^{(t)}=k$, the observed outcome satisfies $Y_i = Y_i^{(t,k)}$.

Assumption \ref{asmp:weak-conditional-exchangeability} imposes unconfoundedness separately for each treatment-version pair $(t,k)$, extending the standard conditional exchangeability assumption to the joint assignment of treatments and versions.
Under these conditions, the target population-level estimands are identifiable \citep{Imbens2000-tm}.
Specifically, conditional on the covariates $\bm{X}_i$, the potential outcome $Y_i^{(t,k)}$ is assumed to be independent of the joint treatment-version indicator $Z_i^{(t,k)}$.
This assumption is related to principal ignorability, which restricts associations between potential outcomes and principal stratum membership conditional on covariates \citep{Stuart2015-pi,Jiang2021-ps}.

In practical applications, Assumption \ref{asmp:weak-conditional-exchangeability} is most plausible when the covariates related to practical constraints and rules, which are determined before assigning the treatment, are observed and included in $\bm{X}_i$.
For instance, when we consider an educational program as the treatment, there exist various types of educational programs: differences in teaching materials, allocation of class time, and whether an aide is present.
This implies that multiple versions exist.
These versions are typically composed of the covariates related to practical constraints and rules (e.g., staffing availability, timetables, and scheduling policies).
Thus, by conditioning on $\bm{X}_i$, the version assignment is roughly randomized with respect to the potential outcomes.

By contrast, when there exist unobserved covariates that simultaneously affect both versions and the potential outcomes, the plausibility of Assumption \ref{asmp:weak-conditional-exchangeability} becomes questionable.
There may exist unobserved covariates, for example, teachers' expertise, motivation, and organizational culture.
When such covariates are not observed, it is difficult to avoid the selection bias.
Addressing such bias would require further assumptions or sensitivity analyses.
These are beyond the scope of this section.
Accordingly, we proceed under Assumptions \ref{asmp:consistency} and \ref{asmp:weak-conditional-exchangeability}.

Specifically, we are often interested in estimating the mean potential outcomes $\mathbb{E}\qty[Y_i^{(t,k)}]$.
By introducing these assumptions, we obtain the following oracle representation if the realized version were observed:
\begin{align*}
  \mathbb{E}\qty[Y_i^{(t, k)} \mid \bm{X}_i] &= \mathbb{E}\qty[Y_i^{(t,k)} \mid Z_i^{(t,k)} = 1, \bm{X}_i]\\
  &= \mathbb{E}\qty[Y_i \mid T_i=t, K_i^{(t)} = k, \bm{X}_i].
\end{align*}
Then, marginalizing over the distribution of $\bm{X}_i$ obtains the mean potential outcome $\psi_{t,k} = \mathbb{E}\qty[Y_i^{(t, k)}]$.
However, estimating the expectation by direct stratification over the covariates is infeasible due to the curse of dimensionality.
Propensity score weighting has become a standard tool in causal inference for addressing this difficulty.
Then, we introduce the treatment and version assignment probabilities as the generalized propensity score \citep{Imbens2000-tm, Imai2004-pf}: 
\begin{definition}[Generalized Propensity Score]
    For each treatment $t \in \mathcal T$, the treatment assignment probability given covariates $\bm X_i$ is defined as
    \begin{align*}
        e_t(\bm{X}_i) &:= \mathrm{Pr} \qty(T_i = t \mid \bm{X}_i) = \mathbb{E}\qty[D_i^{(t)} \mid \bm{X}_i], \ \forall t \in \mathcal{T}.
    \end{align*}
    When the treatment $T_i$ is equal to $t$, the version assignment probability for each $k \in \mathcal{K}^{(t)}$ is defined as
    \begin{align*}
        \rho_{t, k}\qty(\bm{X}_i) &:= \mathrm{Pr}\qty(K_i^{(t)} = k \mid T_i = t, \bm{X}_i)\\
        &= \mathbb{E} \qty[Z_i^{(t, k)} \mid T_i = t, \bm{X}_i], \ \forall k \in \mathcal{K}^{(t)}.
    \end{align*}
    Then, the joint assignment probability for the treatment-version pair $(t,k)$ is given by
    \begin{align*}
        p_{t,k}(\bm X_i) := e_t(\bm{X}_i) \rho_{t, k}\qty(\bm{X}_i).
    \end{align*}
    We refer to $p_{t,k}(\bm X_i)$ as a generalized propensity score for the treatment-version pair $(t, k)$.
\end{definition}
If weak unconfoundedness holds when conditioning on the covariates, then the same property also holds when conditioning on the generalized propensity score.
\begin{proposition}
\label{prop:propensity-score-theorem}
Let
\begin{align*}
 p_{t,k}(\bm X_i) := e_t(\bm X_i)\rho_{t,k}(\bm X_i).
\end{align*}
Under Assumption \ref{asmp:weak-conditional-exchangeability},
for each $(t,k)$,
\begin{align*}
 Y_i^{(t,k)} \indep Z_i^{(t,k)} \mid p_{t,k}(\bm X_i).
\end{align*}
\end{proposition}
The proof is given in Appendix \ref{sec:Proof_of_Proposition1}. 

We next impose a standard overlap condition to ensure that all the treatment-version pairs are probabilistically observable on the covariate space.
\begin{assumption}[Positivity]
    \label{asmp:assumption-Positivity}
    The treatment and version assignment probabilities are positive for all covariate values and treatment-version pairs: for any $(t,k)$, we assume almost surely
      \begin{align*}
      0 < e_t \qty(\bm{X}_i) < 1,
      \quad
      0< \rho_{t,k} \qty(\bm{X}_i) < 1.
      \end{align*}
\end{assumption}
This assumption ensures that every unit has a positive probability of receiving each treatment-version pair within the covariate space characterized by $\bm{X}_i$.
Under Assumptions \ref{asmp:consistency}-\ref{asmp:assumption-Positivity}, the mean potential outcome $\psi_{t,k}$ can be represented by observed data, which is familiar to inverse probability weighting (IPW) representation.
\begin{proposition}[Identification of the mean potential outcome]
\label{prop:identification}
Under Assumptions \ref{asmp:consistency}--\ref{asmp:assumption-Positivity},
the mean potential outcome $\psi_{t,k}$ is identified as 
  \begin{align*}
  \psi_{t,k}
  &=
  \mathbb E\qty[Y_i^{(t,k)}]\\
  &=
  \mathbb E\qty[
    \frac{Z_i^{(t,k)}}
          {e_t(\bm X_i)\rho_{t,k}(\bm X_i)}
    Y_i
  ].
  \end{align*}
\end{proposition}
The proof is given in Appendix \ref{sec:Proof_of_Proposition2}. 

To formalize the causal quantity of interest, we define the pairwise ACE between treatment-version pairs.
For any two pairs $(t, k)$ and $(t', k')$, we define as
\begin{align}
  \Delta_{(t,k), (t',k')} &= \psi_{t',k'} - \psi_{t,k}.
  \label{eq:def-average-causal-effect}
\end{align}
This estimand represents the population-level mean difference in potential outcomes between assigning all units to version $k$ of treatment $t$ and assigning them to version $k'$ of treatment $t'$.
These causal estimands are the primary targets of inference in the remainder of the paper.

The results above establish identification when the treatment version is observed.
The central difficulty considered in this paper arises when the version itself is unobserved.
In Section \ref{sec:proposed-method}, we therefore introduce a finite-mixture representation to recover an identifiable latent structure from the observed data and formulate an additional causal-alignment condition that links the statistical components to substantive versions of treatment.

%
%
\section{Proposed Method}
\label{sec:proposed-method}
In our setting, the treatment $T_i$ is observed but the version $K_i^{(t)}$ is not recorded.
As a result, causal quantities defined in Section \ref{sec:preliminaries} cannot be directly estimated from observed data without additional structure.
Ignoring such treatment heterogeneity induced by multiple versions may obscure the interpretation of the estimated causal effect and lead to ambiguous conclusions \citep{VanderWeele2013-lw}.

To address this problem, we introduce two distinct layers of additional structure.
First, we assume a MoE \citep{Jacobs1991-rx} representation for each observed treatment group.
The latent component in this representation is denoted by $V_i^{(t)} \in \mathcal{V}^{(t)}$ and is referred to as a statistical version.
At this stage, $V_i^{(t)}$ represents only a component of the conditional outcome distribution and does not necessarily correspond to a substantively meaningful implementation of treatment.
Second, we impose the causal-alignment condition given in Assumption \ref{asmp:causal-alignment}.
This condition links each statistical version $v \in \mathcal V^{(t)}$ to one or more substantive versions $k \in \mathcal K^{(t)}$.
Under these conditions, the outcome distribution of MoE is linked to a distribution induced by the corresponding substantive versions.
The conditions enable us to interpret the statistical versions as latent versions.
The following subsections first specify the MoE representation of the observed outcome distribution, then formalize the causal-alignment condition and its implications for the target estimands, and finally develop the maximum-likelihood estimation procedure based on the EM algorithm.

\subsection{Modeling Statistical Versions via Mixture-of-Experts}

For each treatment $t \in \mathcal T$, let $\mathcal V^{(t)}=\{0,1,\ldots,M^{(t)}-1\}$ denote the set of statistical versions, where $M^{(t)} \geq 2$.
We suppose the conditional distribution of the outcome given treatment $t \in \mathcal{T}$ as follows:
\begin{align*}
  f \qty(Y_i \mid T_i = t, \bm{X}_i) = \sum_{v \in \mathcal{V}^{(t)}} \pi_{t,v}(\bm{X}_i) f_{t,v} (Y_i \mid \bm{X}_i),
\end{align*}
where $\pi_{t,v}$ and $f_{t,v}$ denote, respectively, the version assignment model and the outcome model corresponding to treatment $t$ and statistical version $v$.
In the context of an MoE framework, they correspond to the gating function and the expert function, respectively.

Here, we assume that each function represents a parametric model.
The gating function for each treatment $t \in \mathcal{T}$ is modeled by a multinomial logistic regression:
\begin{align}
\pi_{t,v}(\bm{X}_i ; \bm{\eta}_{t})
= \frac{
  \exp \qty(\bm{\eta}_{t,v}^\top \widetilde{\bm X}_i)
}{
  \sum_{u \in \mathcal{V}^{(t)}} \exp\qty( \bm{\eta}_{t, u}^\top \widetilde{\bm X}_i  )
}, \quad \text{ for all } v \in \mathcal{V}^{(t)},
\label{eq:version_assignment_model}
\end{align}
where $\widetilde{\bm X}_i = (1, \bm{X}_i^\top)^\top$ and $\bm{\eta}_{t,v} \in \mathbb{R}^{p+1}$ denotes the coefficient vector for the statistical version $v \in \mathcal{V}^{(t)}$.
For identification, we set $\bm{\eta}_{t,0} = \bm{0}$.
Next, we assume normality for the expert function and specify the model as
\begin{align*}
  Y_i \mid T_i = t, V_i^{(t)} = v , \bm{X}_i \sim \mathcal{N}(\mu_{t,v}(\bm{X}_i), \sigma_{t,v}^2),
\end{align*}
where $\mu_{t,v} (\bm{X}_i) = \bm{\beta}_{t,v}^\top \widetilde{\bm X}_i$.
Here, $\bm{\beta}_{t,v} \in \mathbb{R}^{p+1}$ is the regression coefficient vector and $\sigma_{t,v}>0$ is the standard deviation for each $(t,v)$.
Let $\bm{\gamma}_{t,v} := (\bm{\beta}_{t,v}^{\top}, \sigma_{t,v})^{\top}$ denote the expert parameters, and let $f_{t,v}(y \mid \bm{X}_i; \bm{\gamma}_{t,v})$ be the corresponding probability density function.
Finally, the treatment assignment model is also modeled by a multinomial logistic regression:
\begin{align}
e_{t}(\bm{X}_i ; \bm{\zeta})
= \frac{
  \exp \qty( \bm{\zeta}_{t}^\top \widetilde{\bm X}_i)
}{
  \sum_{u \in \mathcal{T}} \exp\qty( \bm{\zeta}_{u}^\top \widetilde{\bm X}_i )
}, \quad \text{ for all } t \in \mathcal{T},
\label{eq:treatment_assignment_model}
\end{align}
where $\bm{\zeta}_{t} \in \mathbb{R}^{p+1}$ denotes the coefficient vector in the treatment assignment model for each $t \in \mathcal{T}$, and $\bm{\zeta}_{0}$ is fixed at $\bm{0}$ for identification.
For later use, we collect the parameters as $\bm\eta_t:=\qty{\bm\eta_{t,v}:v\in\mathcal V^{(t)}}$, $\bm\gamma_t := \qty{\bm\gamma_{t,v}:v\in\mathcal V^{(t)}}$, $\bm\theta_t:= \qty{\bm\eta_t,\bm\gamma_t}$, $\bm\zeta:=\qty{\bm\zeta_t:t\in\mathcal T}$.

\subsection{Causal Alignment and IPW Representations}
\label{sec:causal-alignment-and-identification}
Motivated by \citet{VanderWeele2022-cm}, we formulate a causal-alignment condition for the MoE representation.
This condition links the statistical versions to the substantive versions defined in Section \ref{sec:preliminaries}.
We then use this condition to derive IPW representations for mean potential outcomes indexed by statistical versions and their pairwise contrasts.

\begin{assumption}[Causal alignment of substantive and statistical versions]
\label{asmp:causal-alignment}
For each treatment $t\in\mathcal T$, there exists a surjective map $h_t:\mathcal K^{(t)}\longrightarrow\mathcal V^{(t)}$ such that: 
\begin{enumerate}[label=(\roman*)]
 \item
 \textit{Version alignment:}
 \begin{align}
  V_i^{(t)} = h_t(K_i^{(t)}) \quad \text{almost surely}.
  \label{eq:alignment-map}
 \end{align}

 \item
\textit{Outcome equivalence:}
 for any $k,k'\in\mathcal K^{(t)}$ such that
 $h_t(k)=h_t(k')$,
 \begin{align*}
  Y_i^{(t,k)} = Y_i^{(t,k')} \quad \text{almost surely}.
 \end{align*}
\end{enumerate}
\end{assumption}

Assumption \ref{asmp:causal-alignment} provides a structural bridge between the substantive implementations of treatment and the latent components of the MoE.
The surjectivity of $h_t$ ensures that every statistical version has at least one substantive version, thereby excluding components with no substantive interpretation.
Condition (ii) is motivated by the idea of treatment variation irrelevance \citep{VanderWeele2009-xa, VanderWeele2013-lw} and requires that the potential outcome be the same substantive versions mapped to the same statistical version for each unit.
Consequently, the statistical versions are a coarsening of the substantive versions through $h_t$.

Assumption \ref{asmp:causal-alignment} provides a structural bridge between the substantive implementations of treatment and the latent components of the MoE.
The surjectivity of $h_t$ ensures that every statistical version has at least one substantive version, thereby excluding components with no substantive interpretation.
Condition (ii) applies the idea of treatment variation irrelevance \citep{VanderWeele2013-lw} to substantive versions mapped to the same statistical version, requiring their potential outcomes to coincide for each unit.
Hence, each statistical version represents a nonempty set of causally equivalent substantive versions.
Consequently, the statistical versions are a coarsening of the substantive versions through $h_t$.

For each $v\in\mathcal V^{(t)}$, let
\begin{align*}
\mathcal K_{t,v}:=h_t^{-1}\qty(\qty{v})
\end{align*}
denote the set of substantive versions mapped to statistical version $v$. By condition (i) of Assumption \ref{asmp:causal-alignment},
\begin{align}
\pi_{t,v}\qty(\bm X_i) &:=\Pr\qty(V_i^{(t)}=v\mid T_i=t,\bm X_i)\nonumber\\
&=\Pr\qty(K_i^{(t)}\in\mathcal K_{t,v}\mid T_i=t,\bm X_i)\nonumber\\
&=\sum_{k\in\mathcal K_{t,v}}\rho_{t,k}\qty(\bm X_i).
\label{eq:alignment-gating}
\end{align}
Thus, the gating probability for statistical version $v$ is the aggregate assignment probability of the substantive versions in $\mathcal K_{t,v}$.
Furthermore, condition (ii) of Assumption \ref{asmp:causal-alignment} allows us to define the potential outcome indexed by statistical version $v$ as
\begin{align}
 Y_i^{(t,v)} := Y_i^{(t,k)}
 \quad \text{for any } k\in\mathcal K_{t,v}.
 \label{eq:component-potential-outcome}
\end{align}
This definition is independent of the choice of substantive version within $\mathcal K_{t,v}$.
Therefore, the causal estimands introduced in Section \ref{sec:preliminaries} can be expressed by using the corresponding statistical version:
\begin{align}
 \psi_{t,v}
 &:= \mathbb E\qty[Y_i^{(t,v)}],
 \nonumber\\
 \Delta_{(t,v),(t',v')}
 &:= \psi_{t',v'}-\psi_{t,v}.
 \label{eq:mean-potential-outcome}
\end{align}
For any $k\in\mathcal K_{t,v}$ and $k'\in\mathcal K_{t',v'}$, we have $\psi_{t,v}=\psi_{t,k}$ and $\Delta_{(t,v),(t',v')}=\Delta_{(t,k),(t',k')}$.
Thus, $\Delta_{(t,v),(t',v')}$ is the average causal effect of the intervention from $(t,k)$ to $(t',k')$.
This effect does not depend on which substantive versions are chosen within the respective sets $\mathcal K_{t,v}$ and $\mathcal K_{t',v'}$.

To derive the IPW representations, we first show that the causal assumptions on substantive versions imply the corresponding properties for statistical versions.
Let $C_i^{(t,v)}:=\mathbf 1\qty{V_i^{(t)}=v}$ denote the indicator function for statistical version $v$.
By the version-alignment condition \eqref{eq:alignment-map}, we have
\begin{align*}
 D_i^{(t)}C_i^{(t,v)}
 =
 \sum_{k\in\mathcal K_{t,v}} Z_i^{(t,k)}.
\end{align*}
When $D_i^{(t)}C_i^{(t,v)}=1$, the version-alignment condition implies that $K_i^{(t)}\in\mathcal K_{t,v}$.
Assumption \ref{asmp:consistency} and \eqref{eq:component-potential-outcome} therefore yield
\begin{align*}
 Y_i=Y_i^{(t,v)}
 \quad\text{when }D_i^{(t)}C_i^{(t,v)}=1.
\end{align*}
Thus, consistency is inherited by the statistical versions.
The following proposition establishes the corresponding inheritance of weak conditional exchangeability.
\begin{proposition}[Weak conditional exchangeability for aligned versions]
\label{prop:component-exchangeability}
Suppose that Assumptions \ref{asmp:weak-conditional-exchangeability} and \ref{asmp:causal-alignment} hold.
For each $t\in\mathcal T$ and $v\in\mathcal V^{(t)}$, it holds
\begin{align*}
 Y_i^{(t,v)}
 \indep \qty(D_i^{(t)}C_i^{(t,v)})
 \mid \bm X_i.
\end{align*}
\end{proposition}
The proof is given in Appendix \ref{sec:Proof_of_component-exchangeability}.
By \eqref{eq:alignment-gating} and Assumption \ref{asmp:assumption-Positivity}, positivity assumption is also inherited by the statistical versions:
\begin{align*}
 \Pr\qty(T_i=t,V_i^{(t)}=v\mid\bm X_i)
 =
 e_t(\bm X_i)\pi_{t,v}(\bm X_i)>0
 \quad \text{almost surely}.
\end{align*}

Applying the IPW representation in Proposition \ref{prop:identification} to the aligned treatment-version pair $(t,v)$ and then using the law of iterated expectations, we obtain
\begin{align}
  \psi_{t,v}
  &=
  \mathbb E\qty[
    \frac{D_i^{(t)}C_i^{(t,v)}Y_i}
         {e_t(\bm X_i)\pi_{t,v}(\bm X_i)}
  ]\nonumber\\
  &=
  \mathbb E\qty[
    \frac{D_i^{(t)}Y_i}
         {e_t(\bm X_i)\pi_{t,v}(\bm X_i)}
    \mathbb E\qty[
      C_i^{(t,v)}
      \mid Y_i,\bm X_i,T_i
    ]
  ]
  \nonumber\\
  &=
  \mathbb E\qty[
    \frac{D_i^{(t)}r_{t,v}(Y_i,\bm X_i)Y_i}
         {e_t(\bm X_i)\pi_{t,v}(\bm X_i)}
  ],
  \label{eq:component-ipw-observed}
\end{align}
where $r_{t,v}(Y_i,\bm X_i):= \Pr\qty(V_i^{(t)}=v \mid T_i=t,Y_i,\bm X_i)$ denotes the posterior probability of statistical version $v$ given treatment $t$, the outcome, and the covariates.
Thus, the unobserved joint treatment-version indicator is replaced by its conditional expectation given the observed data.
In this paper, we use this IPW representation to estimate the pairwise ACE.

By fitting the treatment assignment model and applying the EM algorithm described in Section \ref{sec:estimation}, we can estimate the treatment assignment probability $e_t(\bm X_i)$, the version assignment probability $\pi_{t,v}(\bm X_i)$,
and the posterior probability of the statistical version $r_{t,v}(Y_i,\bm X_i)$.
By using these estimates, the mean potential outcome for each statistical version $\psi_{t,v}=\mathbb E\qty[Y_i^{(t,v)}]$ can be estimated using the IPW representation \eqref{eq:component-ipw-observed}.

We construct the following plug-in Horvitz-Thompson (HT) estimator \citep{Horvitz1952-nw} by replacing the unknown probability functions $e_t(\bm X_i)$, $\pi_{t,v}(\bm X_i)$, and $r_{t,v}(Y_i,\bm X_i)$ with their estimates $e_t(\bm X_i;\widehat{\bm\zeta})$, $\pi_{t,v}(\bm X_i;\widehat{\bm\eta}_t)$, and $r_{t,v}(Y_i,\bm X_i;\widehat{\bm\theta}_t)$, respectively:
\begin{align}
  \widehat\psi_{t,v}^{\mathrm{HT}}
  &:=
  \frac{1}{n}\sum_{i=1}^n
  \frac{
    D_i^{(t)}
    r_{t,v}(Y_i,\bm X_i;\widehat{\bm\theta}_t)
  }{
    e_t(\bm X_i;\widehat{\bm\zeta})
    \pi_{t,v}(\bm X_i;\widehat{\bm\eta}_t)
  }Y_i.
  \label{eq:ht-estimator}
\end{align}
This estimator allows us to estimate the mean potential outcome corresponding to the treatment-version pair $(t,v)$.

We estimate the pairwise ACE defined in \eqref{eq:mean-potential-outcome} by taking the difference between the corresponding version means:
\begin{align}
  \widehat\Delta_{(t,v),(t',v')}
  &=
  \widehat\psi_{t',v'}^{\mathrm{HT}}
  -
  \widehat\psi_{t,v}^{\mathrm{HT}}.
  \label{eq:pairwise-ace-estimator}
\end{align}
Under Assumption \ref{asmp:causal-alignment}, this estimator targets $\Delta_{(t,k),(t',k')}$ for any $k\in\mathcal K_{t,v}$ and $k'\in\mathcal K_{t',v'}$.

%
%
\subsection{Estimation via EM algorithm}
\label{sec:estimation}
Next, we formulate the incomplete-data likelihood for the implementation of the EM algorithm \citep{Dempster1977-yd}.
Since the latent version indicators are unobserved, the likelihood of the observed dataset is obtained by marginalizing over the latent versions.
Specifically, the incomplete-data likelihood is given by
\begin{align}
  L_n(\bm{\zeta},\bm{\theta})
  &= \prod_{i=1}^{n}
  \prod_{t \in \mathcal{T}}
  \qty{
    e_t(\bm{X}_i; \bm{\zeta})
    \sum_{v \in \mathcal{V}^{(t)}}
    \pi_{t,v}(\bm{X}_i; \bm{\eta}_t)
    f_{t,v}(Y_i \mid \bm{X}_i; \bm{\gamma}_{t,v})
  }^{D_i^{(t)}},
  \label{eq:incomplete_likelihood}
\end{align}
where $\bm{\theta}=\qty{\bm{\theta}_t:t\in\mathcal T}$ denotes the collection of MoE parameters, and $(\bm{\zeta},\bm{\theta})$ denotes the full parameter set.
Taking the logarithm of \eqref{eq:incomplete_likelihood}, the incomplete-data log-likelihood is expressed as
\begin{align*}
  \ell_n(\bm{\zeta},\bm{\theta})
  &=
  \ell_n^{\mathrm{trt}}(\bm{\zeta})
  + \sum_{t\in\mathcal T}
    \ell_{n,t}^{\mathrm{obs}}(\bm{\theta}_t),
\end{align*}
where
\begin{align*}
  \ell_n^{\mathrm{trt}}(\bm{\zeta})
  &:=
  \sum_{i=1}^{n}
  \sum_{t\in\mathcal T}
  D_i^{(t)}
  \log e_t(\bm X_i;\bm{\zeta}),\\
  \ell_{n,t}^{\mathrm{obs}}(\bm{\theta}_t)
  &:=
  \sum_{i:T_i=t}
  \log\qty{
    \sum_{v\in\mathcal V^{(t)}}
    \pi_{t,v}(\bm X_i;\bm{\eta}_t)
    f_{t,v}(Y_i\mid\bm X_i;\bm{\gamma}_{t,v})
  }.
\end{align*}
Here, $\ell_n^{\mathrm{trt}}$ is the log-likelihood for the treatment assignment model, and $\ell_{n,t}^{\mathrm{obs}}$ is the log-likelihood for the mixture model within treatment $t$.
To estimate the parameters $(\bm{\zeta},\bm{\theta})$, we solve the following optimization problems:
\begin{align}
  \widehat{\bm{\zeta}}
  &=
  \argmax_{\bm{\zeta}:\bm{\zeta}_0=\bm 0}
  \ell_n^{\mathrm{trt}}(\bm{\zeta}),
  \label{eq:treatment-mle}
\end{align}
and, for each treatment $t$,
\begin{align}
  \widehat{\bm{\theta}}_t
  &=
  \argmax_{\bm{\theta}_t:\bm{\eta}_{t,0}=\bm 0}
  \ell_{n,t}^{\mathrm{obs}}(\bm{\theta}_t).
  \label{eq:incomplete_maximization}
\end{align}
We could not maximize the incomplete-data log-likelihood \eqref{eq:incomplete_maximization} directly because of the summation inside the logarithm.
Therefore, we employ the EM algorithm to replace these latent indicators with their posterior expectations (E-step) and update the parameters $\bm{\theta}_t$ to maximize the expected complete-data log-likelihood (M-step).

Here, we take the conditional expectation of the complete-data log-likelihood with respect to the posterior distribution of the latent version indicators $C_i^{(t,v)}$ under the current parameter estimates $\bm{\theta}_t^{(s)}$, where $s$ denotes the iteration number.
The complete-data log-likelihood for treatment $t$ is given by
\begin{align*}
  \ell_{n,t}^{\mathrm{comp}}(\bm{\theta}_t)
  &= \sum_{i:T_i=t}
  \sum_{v \in \mathcal{V}^{(t)}}
  C_i^{(t,v)}
  \qty{
    \log \pi_{t,v}(\bm{X}_i; \bm{\eta}_t)
    + \log f_{t,v}(Y_i \mid \bm{X}_i; \bm{\gamma}_{t,v})
  }.
\end{align*}
Then, the expected complete-data log-likelihood for treatment $t$ is computed as
\begin{align*}
  &Q_t(\bm{\theta}_t \mid \bm{\theta}_t^{(s)})
  \\
  &=
  \mathbb E_{\bm{\theta}_t^{(s)}}\qty[
    \ell_{n,t}^{\mathrm{comp}}(\bm{\theta}_t)
    \mid \qty{(Y_i,\bm X_i,T_i)}_{i=1}^n
  ]
  \\
  &=
  \sum_{i:T_i=t}
  \sum_{v \in \mathcal{V}^{(t)}}
  \mathbb E_{\bm{\theta}_t^{(s)}}\qty[
    C_i^{(t,v)}
    \mid T_i=t,Y_i,\bm X_i
  ]
  \qty{
    \log \pi_{t,v}(\bm{X}_i; \bm{\eta}_t)
    + \log f_{t,v}(Y_i \mid \bm{X}_i; \bm{\gamma}_{t,v})
  }
  \\
  &=
  \sum_{i:T_i=t}
  \sum_{v \in \mathcal{V}^{(t)}}
  r_{i,t,v}^{(s)}
  \qty{
    \log \pi_{t,v}(\bm{X}_i; \bm{\eta}_t)
    + \log f_{t,v}(Y_i \mid \bm{X}_i; \bm{\gamma}_{t,v})
  }.
\end{align*}
Here, $r_{i,t,v}^{(s)}$ denotes the posterior responsibility at iteration $s$, which is the conditional probability that observation $i$ belongs to latent version $v$ given treatment $t$, the outcome, and the covariates.
Applying Bayes' theorem to the posterior probability and evaluating it at $\bm\theta_t^{(s)}$ gives
\begin{align}
  r_{i,t,v}^{(s)}
  &:=
  r_{t,v}(Y_i,\bm X_i;\bm{\theta}_t^{(s)})
  \nonumber\\
  &=
  \frac{
    \pi_{t,v}(\bm X_i;\bm{\eta}_t^{(s)})
    f_{t,v}(Y_i\mid\bm X_i;\bm{\gamma}_{t,v}^{(s)})
  }{
    \displaystyle\sum_{u\in\mathcal V^{(t)}}
    \pi_{t,u}(\bm X_i;\bm{\eta}_t^{(s)})
    f_{t,u}(Y_i\mid\bm X_i;\bm{\gamma}_{t,u}^{(s)})
  }.
  \label{eq:estep}
\end{align}
In the E-step, $r_{t,v}(Y_i,\bm X_i;\bm{\theta}_t^{(s)})$ is computed at the current parameter estimates.
In the M-step, the parameters are updated by maximizing the expected complete-data log-likelihood for treatment $t\in\mathcal T$:
\begin{align}
  \bm{\eta}_t^{(s+1)}
  &=
  \argmax_{\bm{\eta}_t: \bm{\eta}_{t,0}=\bm 0}
  \sum_{i:T_i=t}
  \sum_{v\in\mathcal V^{(t)}}
  r_{i,t,v}^{(s)}
  \log\pi_{t,v}(\bm X_i;\bm{\eta}_t),
  \label{eq:mstep-gating}\\
  \bm{\gamma}_{t,v}^{(s+1)}
  &=
  \argmax_{\bm{\gamma}_{t,v}}
  \sum_{i:T_i=t}
  r_{i,t,v}^{(s)}
  \log f_{t,v}(Y_i\mid\bm X_i;\bm{\gamma}_{t,v}),
  \quad v\in\mathcal V^{(t)}.
  \label{eq:mstep-expert}
\end{align}
The gating parameters are updated jointly over all versions within treatment $t$, whereas the expert parameters are updated separately for each version.
The details of updating the parameters for the outcome model are provided in Appendix \ref{sec:closed-form-solutions-for-gaussian-outcome-experts}.
Algorithm \ref{alg:estimation-procedure} summarizes the estimation procedure.
\begin{algorithm*}
\caption{Estimation of mean potential outcomes for aligned version and their pairwise ACEs}
\label{alg:estimation-procedure}
\begin{algorithmic}[1]
\Require Data $\{(Y_i,\bm X_i,T_i)\}_{i=1}^n$, the set of latent versions $\{M^{(t)}\}_{t\in\mathcal T}$.
\Ensure Mean potential outcomes $\widehat\psi_{t,v}^{\mathrm{HT}}$, $\widehat\psi_{t',v'}^{\mathrm{HT}}$ and pairwise ACEs $\widehat\Delta_{(t,v),(t',v')}$.

\State Fit the treatment assignment model by \eqref{eq:treatment-mle}, with $\bm\zeta_0=\bm 0$.

\For{$t\in\mathcal T$}
  \For{each initialization}
    \State Initialize $\bm\theta_t^{(0)}$ by random values, and fix $\bm\eta_{t,0}^{(0)}=\bm 0$, and set $s\gets 0$.
    \Repeat
      \State \textbf{E-step:} Compute $r_{i,t,v}^{(s)}$ by \eqref{eq:estep} for all $i$ with $T_i=t$ and $v\in\mathcal V^{(t)}$.
      \State \textbf{M-step:} Update the gating parameters by \eqref{eq:mstep-gating} and the expert parameters by \eqref{eq:mstep-expert}.
      \State $s\gets s+1$.
    \Until{
    $|\ell_{n,t}^{\mathrm{obs}}(\bm\theta_t^{(s)})
    -\ell_{n,t}^{\mathrm{obs}}(\bm\theta_t^{(s-1)})|
    <\mathrm{tol}$
    }
  \EndFor
  \State Order the components and reset the gating baseline as described in Section \ref{sec:identifiability-and-consistency}.
\EndFor

\State Compute the fitted assignment probabilities and posterior responsibilities.
\State Compute $\widehat\psi_{t,v}^{\mathrm{HT}}$ by \eqref{eq:ht-estimator} and $\widehat\Delta_{(t,v),(t',v')}$ by \eqref{eq:pairwise-ace-estimator}.
\end{algorithmic}
\end{algorithm*}

\subsection{Identifiability of the MoE and Asymptotic Properties of the IPW Estimator}
\label{sec:identifiability-and-consistency}

When constructing estimators based on mixture models, it is essential to examine the identifiability of the parameters.  
In particular, the MoE model, which consists of both gating and expert functions, is known to involve more complex identifiability conditions than standard finite mixture models.  
\citet{Jiang1999-vw} systematically established sufficient conditions for the identifiability of MoE parameters and proposed methodological procedures to achieve unique parameter representations by resolving permutational invariance and translational invariance.  
In this study, we follow their theoretical framework and ensure model identifiability through similar procedures.
Throughout this subsection, the number of components $M^{(t)}$ is fixed, and the expert parameter vectors $\bm\gamma_{t,v}$ are mutually distinct within each treatment.

The main challenges to identifiability in the MoE model are as follows:
\begin{enumerate}
  \item Permutational invariance: the labeling ambiguity of expert components.
  \item Translational invariance: non-unique parameterization of the gating functions due to additive shifts in the coefficients.
\end{enumerate}

To address these invariances, we estimate the parameters $\bm{\theta}_{t}$ along with the following steps:

\begin{enumerate}
  \item Estimate parameters $\bm{\theta}_{t}$ by fixing the coefficient vector of the reference gating function to the zero vector $\bm{\eta}_{t,0} = \bm{0}$ for each $t\in\mathcal{T}$.
  \item After estimating parameters $\bm{\theta}_t$ by EM algorithm for each $t\in\mathcal{T}$, reorder the expert components by a deterministic lexicographic rule based on $\qty{\bm{\gamma}_{t,v}}_{v \in \mathcal{V}^{(t)}}$.
  \item After reordering, subtract the gating coefficient vector of the new reference component from every gating coefficient vector so that $\bm\eta_{t,0}=\bm 0$.
\end{enumerate}
In the second step, the estimated expert components characterized by parameters $\bm{\gamma}_t$ for each treatment $t \in \mathcal{T}$ are sorted according to the following lexicographic rule:  
for any pair of components $j$ and $k$, we define $f_{t,j} \prec f_{t,k}$ if and only if
there exists the smallest index $\ell$ such that $\gamma_{t,j,\ell} \neq \gamma_{t,k,\ell}$ and
$\gamma_{t,j,\ell} < \gamma_{t,k,\ell}$, where $\gamma_{t,j,l}$ denotes the $l$-th element of the coefficient vector $\bm{\gamma}_{t,j}$.

Under these conditions, the fixing, ordering, and re-fixing procedures yield a unique parameter representation by Theorem 1 and Lemma 3(b) of \citet{Jiang1999-vw}, corresponding to identifiability in the sense of their Definition 2.
When the identifiability is ensured, the standard large-sample properties of the MoE estimator can be established under the standard regularity conditions commonly assumed in finite mixture models.

Based on these identifiability results, we can establish the asymptotics of the proposed IPW estimator \eqref{eq:ht-estimator}.
For simplicity of notation, for fixed $(t,v)$, define
\begin{align*}
  g_i(\bm\zeta,\bm\theta_t)
  :=
  \frac{
    D_i^{(t)}r_{t,v}(Y_i,\bm X_i;\bm\theta_t)
  }{
    e_t(\bm X_i;\bm\zeta)
    \pi_{t,v}(\bm X_i;\bm\eta_t)
  }Y_i.
\end{align*}
The plug-in Horvitz-Thompson estimator can be expressed as
\begin{align*}
  \widehat\psi_{t,v}^{\mathrm{HT}}
  =
  \frac{1}{n}\sum_{i=1}^n
  g_i(\widehat{\bm\zeta},\widehat{\bm\theta}_t).
\end{align*}
We provide the following theorem, which establishes the consistency of the proposed IPW estimator.

\begin{theorem}[Consistency]
\label{thm:consistency}
Fix a treatment $t$ and a statistical version $v$.
Suppose that
\begin{enumerate}[label=(\arabic*)]
  \item Assumptions \ref{asmp:consistency}--\ref{asmp:causal-alignment} hold.
  \item The conditions of Lemma \ref{lemma:consistency_model_parameters} hold, and the parameter space $\Theta$ for $(\bm\zeta,\bm\theta_t)$ is compact.
  \item $\mathbb E\qty[|Y_i|]<\infty$.
  \item The function $g_i(\bm\zeta,\bm\theta_t)$ is continuous on $\Theta$ almost surely, and there exists a random variable $M_i$ such that
  \begin{align*}
    \sup_{(\bm\zeta,\bm\theta_t)\in\Theta}
    |g_i(\bm\zeta,\bm\theta_t)|
    \le M_i,
    \qquad
    \mathbb E\qty[M_i]<\infty.
  \end{align*}
\end{enumerate}
Then,
\begin{align*}
  \widehat\psi_{t,v}^{\mathrm{HT}}
  \xrightarrow{p}\psi_{t,v}.
\end{align*}
\end{theorem}
The proof is given in Appendix \ref{sec:Proof_of_Theorem1}. 
By the continuous mapping theorem, for any two treatment-version pairs satisfying the conditions of Theorem \ref{thm:consistency},
\begin{align*}
  \widehat\Delta_{(t,v),(t',v')}
  &=
  \widehat\psi_{t',v'}^{\mathrm{HT}}
  -
  \widehat\psi_{t,v}^{\mathrm{HT}}
  \xrightarrow{p}
  \Delta_{(t,v),(t',v')}.
\end{align*}

%
%
\section{Simulation Study}
\label{sec:simulation-study}

We evaluate the effectiveness of the proposed method via Monte Carlo simulations.
Multiple datasets are generated according to the models \eqref{eq:version_assignment_model} and \eqref{eq:treatment_assignment_model} in our simulations.

\subsection{Setup}

For each replication of the dataset, we first generated $n$ independent observations $\qty{\bm{X}_i}_{i=1}^n$ from a $p$-dimensional standard normal distribution $\mathcal{N}(\bm{0}, I_p)$.
We then conducted simulations under correctly specified models, as described below.
The data-generating process represents the bijective case of Assumption \ref{asmp:causal-alignment} by construction.
Each generated version defines one distinct MoE expert.
The simulation therefore evaluates estimation under correct causal alignment.
It does not empirically validate the causal alignment assumption.

The treatment $T_i$ was generated according to the following treatment assignment model:
\begin{align*}
e_{t}(\bm{X}_i ; \bm{\zeta}^*)
= \frac{
  \exp \qty( {\bm{\zeta}_{t}^*}^\top \bm{X}_i)
}{
  \sum_{u \in \mathcal{T}} \exp\qty( {\bm{\zeta}_{u}^*}^\top \bm{X}_i )
}, \quad \text{ for all } t \in \mathcal{T}, 
\end{align*}
where $\bm{\zeta}_t^*$ is a true parameter for all $t \in \mathcal{T}$. 
Then, conditional on $T_i=t$, we assigned the latent version $V_i^{(t)}$ according to the version assignment model:
\begin{align*}
\pi_{t,v}(\bm{X}_i ; \bm{\eta}_{t}^*)
= \frac{
  \exp \qty({\bm{\eta}_{t,v}^*}^\top \bm{X}_i)
}{
  \sum_{u \in \mathcal{V}^{(t)}} \exp \qty( {\bm{\eta}_{t, u}^*}^\top \bm{X}_i  )
}, \quad \text{ for all } v \in \mathcal{V}^{(t)},
\end{align*}
where $\bm{\eta}_{t,v}^*$ is the true parameter for all $v \in \mathcal{V}^{(t)}$.
Both the treatment assignment model $e_t(\bm{X}_i ; \bm{\zeta}^*)$ and the version assignment model $\pi_{t,v}(\bm{X}_i ; \bm{\eta}_{t}^*)$ were correctly specified settings.
Finally, we generated the outcome $Y_i$ from the following model:
\begin{align*}
    Y_i &= \sum_{t \in \mathcal{T}} \sum_{v \in \mathcal{V}^{(t)}}  D_i^{(t)} C_i^{(t, v)} Y_i^{(t,v)},\\
      \tilde{Y}_i
  &= \sum_{t \in \mathcal{T}}
     \sum_{v \in \mathcal{V}^{(t)}}
     D_i^{(t)} C_i^{(t,v)}
     \qty(\bm{X}_i^\top \bm{\beta}_{t,v}^*
     + \psi_{t,v}^*),\\
    Y_i^{(t,v)} &= \bm{X}_i^\top \bm{\beta}_{t,v}^* + \psi_{t,v}^* + \varepsilon_i, \quad \varepsilon_i \sim \mathcal{N}(0, \sigma^2),
\end{align*}
where $\varepsilon_i$ is an error term, $\bm{\beta}_{t,v}^*$ is the true parameter for the outcome model, and $\psi_{t,v}^*$ is the true mean potential outcome when the observation is treated with the treatment $t$ and the latent version $v$.
To fix a prespecified signal-to-noise ratio (SNR), we set the noise scale $\sigma = \sqrt{\mathbb{V}(\tilde{Y}_i)}/\mathrm{SNR}$.

We considered the true parameters for the treatment, version assignment models, and the outcome model.
We modeled the multinomial logistic model for $T_i$.
The coefficients $\bm{\zeta}_0$ were set to $\bm{0}$.
For each treatment $t = 1, \dots, J-1$, we set exactly two coordinates of $\bm{\zeta}_t$ with opposite signs.
Using cyclic indices $j_{1(t)} = 1 + \qty(2t - 2 \bmod p), j_{2(t)} = 1 + \qty(2t - 1 \bmod p)$, we set $\qty(\bm{\zeta}_t^*)_{j_{1(t)}} = 2.0, \qty(\bm{\zeta}_t^*)_{j_{2(t)}} = -2.0$, and all remaining entries to zero.

For each treatment $t$, we used a multinomial logistic model with latent version $v=0$ having zero coefficients: $\bm{\eta}_{t, 0}^* = \bm{0}$.
To ensure that the version assignment model primarily uses covariate coordinates disjoint from those that was used in the treatment assignment model, we introduced an offset $2(J-1)$, corresponding to the $2(J-1)$ coordinates already used to distinguish treatments $t = 1, \dots, J-1$.
For each latent version $v = 1, \dots, M^{(t)}-1$, we also set two coordinates of $\bm{\eta}_{t,v}$.
Using cyclic indices $j_{1(t,v)} = 1 + \text{offset} + (2v - 2 \bmod p), j_{2(t,v)} = 1 + \text{offset} + (2v - 1 \bmod p),$ we set $\qty(\bm{\eta}_{t,v}^*)_{j_1(t,v)} = 2.0, \qty(\bm{\eta}_{t,v}^*)_{j_2(t,v)} = 2.0$ and all remaining entries to zero.

For the true coefficient parameter of outcomes, we defined a common covariate effect shared across all the treatments and latent versions as the unit vector:
$$
\bm{\beta}_{\text{common}}
= \frac{1}{\sqrt{p}} (1,1,\dots,1)^\top.
$$
To induce identifiability of the MoE, we introduced small version-specific perturbations to $\bm{\beta}_{\text{common}}$.
Let $\qty{\bm{b}_0, \dots, \bm{b}_{p-1}}$ denote the standard basis of $\mathbb R^p$.
For each treatment-version pair $(t,v)$, we defined $ \bm{\beta}_{t,v}^*
= \bm{\beta}_{\text{common}} + 0.2 \bm{b}_v$.
Moreover, the outcome model included the true mean potential outcome $\psi_{t,v}^*$, which are assigned in lexicographic order over $(t,v)$.
The value of $\psi_{t,v}^*$ starts at $1.0$ and increases by $1.0$ for each subsequent pair.
For example, when $M^{(t)}=2$ for every treatment,
the assigned values are:
\begin{align*}
    \psi_{0,0}^* = 1.0, \quad
    \psi_{0,1}^* = 2.0, \quad
    \psi_{1,0}^* = 3.0, \quad
    \psi_{1,1}^* = 4.0, \quad
    \psi_{2,0}^* = 5.0, \quad
    \psi_{2,1}^* = 6.0, \quad \dots
\end{align*}

\subsection{Performance Evaluation}

We conducted $M = 100$ simulations under some configurations of the ground-truth data generating process.
Across all the configurations, we varied the sample size $n \in \qty{500, 1000, 2000}$, the covariate dimension $p \in \qty{10, 20}$, and a signal-to-noise ratio $\mathrm{SNR} \in \qty{5, 10}$.
We further considered the number of treatments and latent versions.
We set the number of treatment as $J \in \qty{2, 3}$.
For each value of $J$, we defined the latent version structure $(M^{(0)}, M^{(1)}, \dots)$ accordingly.
When $J=2$, the latent version structure was given by $(M^{(0)}, M^{(1)}) \in \qty{(2,2), (3,3)}$, while we considered $(M^{(0)}, M^{(1)}, M^{(2)}) \in \qty{(2,2,2), (3,3,3)}$ when $J=3$.

To connect the causal estimands, we considered the mean potential outcome under the treatment-version pair $(t, v)$ as
\begin{align*}
  \mathbb{E}\qty[Y_i^{(t,v)}]
  &= \mathbb{E}\qty[\bm{X}_i]^\top
     \bm{\beta}_{t,v}^*
     + \psi_{t,v}^*
     + \mathbb{E}\qty[\varepsilon_i]\\
  &= \psi_{t,v}^*.
\end{align*}
Therefore, the true pairwise ACE between two latent versions is $\Delta_{(t,v),(t,v')}^* = \psi_{t,v'}^* - \psi_{t,v}^*$.

For evaluating estimation accuracy, we used the pairwise ACE defined as \eqref{eq:def-average-causal-effect}.
The pairwise ACE can be calculated by the differences of Horvitz-Thompson estimates as $\widehat{\Delta}_{(t, v), (t, v')}^{\text{HT}, (m)} = \widehat{\psi}^{\text{HT}, (m)}_{t, v'} - \widehat{\psi}^{\text{HT}, (m)}_{t,v}$ for any two treatment-version pairs $(t, v)$ and $(t, v')$.
Here, the superscript of $(m)$ means the estimates at the $m$-th simulation.
In the simulation studies, we reported the bias and standard deviation (SD) for any two treatment-version pairs $(t, v)$ and $(t, v')$, given by
\begin{align*}
\mathrm{Bias}
&=
\frac{1}{M}
\sum_{m=1}^M
\qty(
\widehat{\Delta}^{\text{HT},(m)}_{(t,v), (t, v')}
- \Delta_{(t,v), (t, v')}^*
),\\
\mathrm{SD} &=  
\qty[
\frac{1}{M}
\sum_{m=1}^M
\qty(\widehat{\Delta}^{\text{HT},(m)}_{(t,v), (t, v')} - \overline{\Delta}^{\text{HT}}_{(t,v), (t, v')})^2
]^{\frac{1}{2}},\\
\overline{\Delta}^{\text{HT}}_{(t,v), (t, v')}
&=
\frac{1}{M}
\sum_{m=1}^M
\widehat{\Delta}^{\text{HT},(m)}_{(t,v), (t, v')}.
\end{align*}
These evaluators clarify whether our proposed method works well under the correctly specified model.

The simulation results for the setting with $J=2$ and latent version structure $(2,2)$ are reported in Table  \ref{tab:MonteCarloSimulations}, and the results are visualized as boxplots in Figure \ref{fig:MonteCarloSimulations}.
Since the remaining simulation settings exhibit qualitatively similar trends, their detailed results are summarized in the supplementary materials S1.

As shown in Table \ref{tab:MonteCarloSimulations} and Figure \ref{fig:MonteCarloSimulations}, the estimates of the proposed method are close to the true values across all scenarios considered.
Moreover, for each treatment-version contrast between $(t,v)$ and $(t,v')$, increasing the sample size $n$ leads to reducing bias and standard errors.
These results indicate that the proposed method has the expected asymptotic behavior with respect to sample size.
Increasing the SNR also reduces the bias. It implies that clearer separation of latent versions improves the estimation accuracy of the pairwise ACE.
In contrast, a larger SNR does not yield a uniform reduction of SD across all settings.
Finally, when the covariate dimension becomes larger, the finite-sample bias tends to increase relative to low-dimensional settings.
These results suggest that stable estimation under more complex situations typically requires a larger sample size and a sufficiently high SNR.

Consequently, these results demonstrate that the proposed method exhibits stable estimation performance even in finite samples and delivers highly accurate estimates when the latent version structure is sufficiently identifiable.

\begin{table*}[htbp]
\centering
\caption{Results of Monte Carlo simulations for $J=2$ with latent version structure $(2,2)$.}
\label{tab:MonteCarloSimulations}
\begin{tabular}{ll l ll ll}
\toprule
& &
& \multicolumn{2}{c}{$p=10$} & \multicolumn{2}{c}{$p=20$} \\
\cmidrule(lr){4-5}\cmidrule(lr){6-7}
SNR & $n$ & $(t,v),(t,v')$  & Bias & SD & Bias & SD \\
\midrule
\multirow{6}{*}{5}
& \multirow{2}{*}{500}  & (0,0),(0,1) & --0.502616 & 0.453546 & --0.843029 & 0.362044 \\
&                       & (1,0),(1,1) & --0.496383 & 0.491002 & --0.872900 & 0.651796 \\
& \multirow{2}{*}{1000} & (0,0),(0,1) & --0.317267 & 0.420311 & --0.584633 & 0.400126 \\
&                       & (1,0),(1,1) & --0.365345 & 0.440226 & --0.640676 & 0.396309 \\
& \multirow{2}{*}{2000} & (0,0),(0,1) & --0.157735 & 0.277096 & --0.306299 & 0.332593 \\
&                       & (1,0),(1,1) & --0.208446 & 0.332872 & --0.421939 & 0.383613 \\

\midrule

\multirow{6}{*}{10}
& \multirow{2}{*}{500}  & (0,0),(0,1) & --0.184443 & 0.415674 & --0.632034 & 0.505997 \\
&                       & (1,0),(1,1) & --0.213792 & 0.567841 & --0.472874 & 0.608518 \\
& \multirow{2}{*}{1000} & (0,0),(0,1) & --0.065192 & 0.287996 & --0.195369 & 0.272398 \\
&                       & (1,0),(1,1) & --0.149397 & 0.461156 & --0.308065 & 0.499020 \\
& \multirow{2}{*}{2000} & (0,0),(0,1) & --0.085327 & 0.266137 & --0.150087 & 0.202736 \\
&                       & (1,0),(1,1) & --0.111844 & 0.293243 & --0.220190 & 0.316836 \\

\bottomrule
\end{tabular}
\end{table*}

\begin{figure}[htbp]
    \centering
    \includegraphics[keepaspectratio, width=\columnwidth]{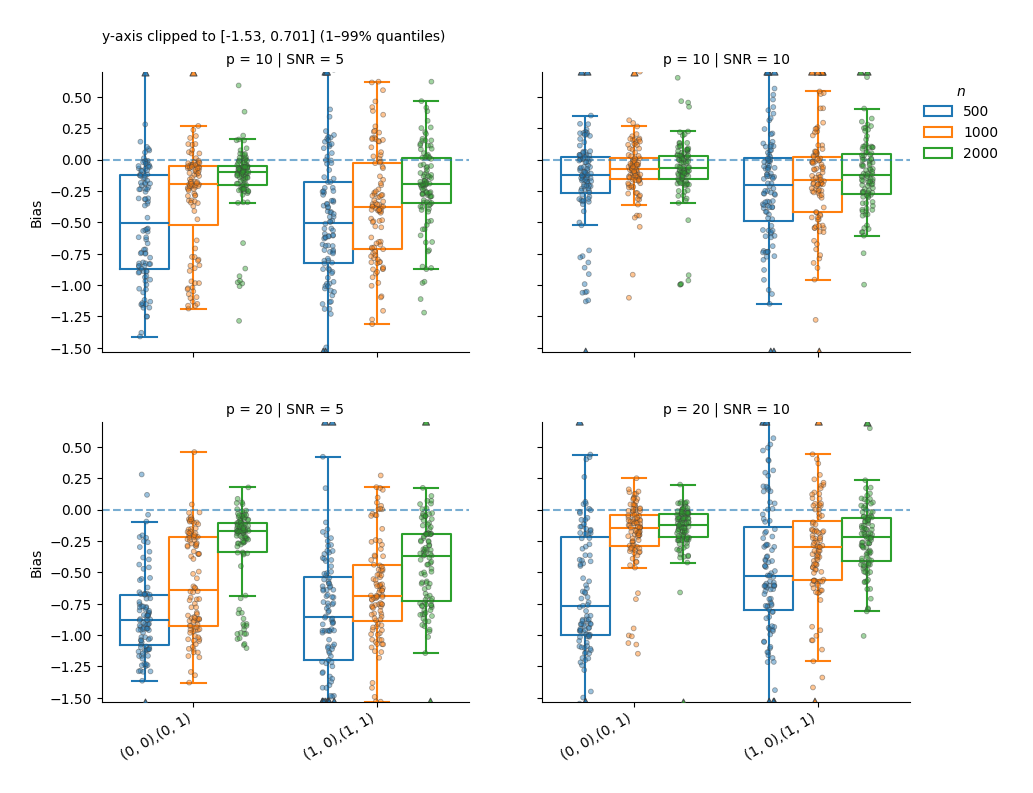}
    \caption{
      Distribution of the bias across Monte Carlo replications.
      For two treatment-version pairs $(t, v)$ and $(t, v')$, the boxplots summarize the bias stratified by sample size $\{500, 1000, 2000\}$, SNR $\{5, 10\}$, and the covariate dimension $\{ 10, 20 \}$.
      Points represent individual Monte Carlo replications. The dashed horizontal line indicates zero bias.
      This plot corresponds to the configuration $J=2$ with latent version structure $(2,2)$.
    }
    \label{fig:MonteCarloSimulations}
\end{figure}

%
%
\section{Empirical Analysis}
\label{sec:empirical-analysis}

\subsection{Dataset Description}

Our empirical analysis demonstrates the effectiveness of our proposed method by applying the Project STAR (Student-Teacher Achievement Ratio) dataset, referred to as the \texttt{STAR} dataset in the CRAN package \texttt{AER} \citep{Kleiber2008-qu}.
The Project STAR dataset is a randomized controlled trial and is widely used in education and causal inference research.
It was conducted in the state of Tennessee, United States, during the late 1980s to evaluate the causal effect of class size on students' academic performance.
The longitudinal survey began in 1985 and followed approximately 7,000 students in 79 schools.
Each student was randomly assigned to one of three treatment conditions in kindergarten:
\begin{enumerate}
  \item \textbf{Small class:} 13--17 students per teacher,
  \item \textbf{Regular class:} 22--25 students per teacher,
  \item \textbf{Regular class with aide:} a regular-size class assisted by a full-time teacher's aide.
\end{enumerate}
Their students were followed from kindergarten through the third grade.
The teachers were also randomly assigned to classrooms, thereby minimizing potential confounding between treatment assignment and teacher characteristics.
The Project STAR dataset contains 11,598 observations and 47 variables.
For each student, the data include background characteristics and mathematics scores from the Stanford Achievement Test.
Covariates of teacher and school are also provided with some characteristics.
Table \ref{tab:star_variables} summarizes the main variables used in the analysis.

Recent studies have revisited the experiment and noted that the realized class sizes often varied across schools even within the same treatment \citep{Adusumilli2024-iv}.
For instance, while the regular class condition was designed to include 22-25 students per teacher, the actual class sizes were determined by some reasons such as schoolroom capacity constraints and attrition from participating schools.
These findings motivate an analysis of heterogeneity within nominal class assignments.
The variables used here do not directly measure instructional practices or classroom processes, so we cannot assign substantive implementation labels to the estimated latent versions based on these data alone.

In light of the discussion in \citet{Adusumilli2024-iv}, our empirical analysis focuses on the subsamples at a kindergarten.
This restriction is based on the fact that many students switched class types due to attritions, behavioral issues, or parental complaints in later grades.
Such noncompliance and reassignment complicate causal identification of the original randomized intervention.

In our analysis restricted to kindergarten as described below, we applied the following preprocessing steps to the Project STAR dataset.
First, missing values in the covariates, treatment, and outcome were excluded, and the analysis was conducted using complete cases only.
Next, we computed category proportions separately within each treatment group for categorical covariates.
We labeled a category as rare when its empirical proportion was less than or equal to $5\%$ in at least one treatment.
Observations containing such rare categories were removed from the dataset.
This procedure was implemented to avoid the estimation instability arising from extremely sparse categories within treatment, such as a singularity of the design matrix.
We standardized numerical covariates to mean zero and unit variance.
Finally, categorical covariates were transformed using one-hot encoding with one reference category omitted for each variable.
The sample of final analysis consisted of 4,600 kindergarten students.

\begin{table*}[htbp]
\centering
\caption{Summary of used variables in our analysis.}
\label{tab:star_variables}
\begin{tabularx}{\textwidth}{p{0.25\textwidth} p{0.25\textwidth} X}
\toprule
\textbf{Category} & \textbf{Variables} & \textbf{Description} \\
\midrule
Student characteristics 
& \texttt{gender}, \texttt{ethnicity}, \texttt{lunchk} 
& Gender, ethnicity, and eligibility for free lunch\\

Academic performance (Outcome)
& \texttt{mathk} 
& Mathematics scores from the Stanford Achievement Test by year\\

Treatment assignment (Treatment)
& \texttt{stark} 
& Class type indicator (small class, regular class, or regular class with aide)\\

Teacher characteristics 
& \texttt{degreek}, \texttt{ladderk}, \texttt{experiencek}, \texttt{tethnicityk} 
& Teacher's degree level, career ladder position, years of teaching experience, and ethnicity\\

School characteristics 
& \texttt{schoolk}
& School location type (inner-city, rural, suburban, urban)\\
\bottomrule
\end{tabularx}
\end{table*}

\subsection{Empirical Results}
By applying the proposed method to the subsamples of kindergarten, we aim to recover the pairwise ACE that arises from unobserved variations.
In the empirical analysis, we consider two candidates of the latent version structure,
\begin{align*}
  \bigl(J^{(\text{small})}, J^{(\text{regular})}, J^{(\text{regular+aide})}\bigr)
\in \{(2,2,2), (3,3,3)\}.
\end{align*}
We generated 100 bootstrap resamples and estimated $\psi_{t,v}$ by using our proposed method.
For each structure, we report nonparametric bootstrap $95\%$ confidence intervals for all the estimates to quantify uncertainty.
The results of the estimated $\psi_{t,v}$ are summarized in Tables \ref{tab:empirical-analysis-J_222} and \ref{tab:empirical-analysis-J_333}.
To summarize the uncertainty in the estimated coefficients of the gating models, we report bootstrap $95\%$ confidence intervals in the form of forest plots, shown in Figures \ref{fig:empirical-analysis-222} and \ref{fig:empirical-analysis-333}. 
The additional results are provided in the supplementary materials S2.

Tables \ref{tab:empirical-analysis-J_222} and \ref{tab:empirical-analysis-J_333} show that the estimated mean potential outcome differs across latent versions and that the \texttt{small} class has higher estimated average outcomes than \texttt{regular} or \texttt{regular+aide}.
These results indicate that unobserved heterogeneity in the implementation of treatment may exist in the treatments.
Note that, in all cases, the overlap of the $95\%$ confidence intervals for the version-specific estimates indicates finite-sample uncertainty.
Accordingly, the results do not provide the conclusion of the effect of the \texttt{small} class differs significantly across latent versions.
A heterogeneity we often overlook in the average treatment effect becomes apparent when latent versions are explicitly considered.

\begin{table}[htbp]
\centering
\caption{IPW estimates with bootstrap 95\% CI under the latent version structure $(2, 2, 2)$.}
\label{tab:empirical-analysis-J_222}
\begin{tabular}{lccc}
\toprule
Treatment & version 0 & version 1 \\
\midrule
small & 482.250 [457.570, 511.973] & 503.483 [467.004, 582.351] \\
regular & 476.445 [452.616, 501.116] & 471.390 [446.781, 506.760] \\
regular+aide & 478.822 [456.566, 521.112] & 488.744 [445.477, 541.403] \\
\bottomrule
\end{tabular}
\end{table}

\begin{table*}[htbp]
\centering
\caption{IPW estimates with bootstrap 95\% CI under the latent version structure $(3, 3, 3)$.}
\label{tab:empirical-analysis-J_333}
\begin{tabular}{lcccc}
\toprule
Treatment & version 0 & version 1 & version 2 \\
\midrule
small & 477.546 [452.736, 513.529] & 487.787 [451.849, 530.566] & 497.456 [452.715, 542.580] \\
regular & 467.594 [444.109, 504.768] & 475.133 [444.235, 521.460] & 470.785 [440.072, 516.305] \\
regular+aide & 473.807 [454.461, 512.279] & 482.914 [451.616, 516.049] & 480.961  [445.563, 515.144] \\
\bottomrule
\end{tabular}
\end{table*}

Next, Figures \ref{fig:empirical-analysis-222} and \ref{fig:empirical-analysis-333} show how latent version assignments relate to observed covariates.
Under both the latent version structures $(2,2,2)$ and $(3,3,3)$, the estimated coefficients associated with latent versions $1$ and $2$ are generally concentrated near zero.
The coefficients for latent versions $1$ and $2$ are similar in the latent version structure $(3,3,3)$.
These results suggest that assignment to latent versions is not strongly explained by the observed covariates alone, at least for the present data and model specifications.
Moreover, little differences between latent versions $1$ and $2$ indicate that the number of latent versions may be overspecified.

On the other hand, interpreting the results based on the average tendencies of the estimated coefficients in the gating model reveals consistency with existing research.
In the \texttt{small} class, latent versions $1$ and $2$ tend to be more frequently assigned to classrooms characterized by teachers with a bachelors' degree, fewer years of teaching experiences, female students, and apprentices in their career level.
The corresponding IPW estimates indicate that the average math score of these versions are higher than that of latent version $0$.
Taken together, these tendencies suggest that latent versions $1$ and $2$ correspond to a restricted educational environment, and can be interpreted as the implementation of the treatment which achieves a higher math score.
This interpretation is consistent with existing evidence from Project STAR.
In particular, \citet[p.\,2]{Adusumilli2024-iv} report that ``Nearly all of the gains from reducing class size in STAR are driven by $29\%$ of the schools in the sample.''
This finding implies that the effectiveness of the reduction of class size may be concentrated in a subset of schools.
From this perspective, our version-specific findings can be viewed as complementary to and consistent with prior empirical results.

\begin{figure}[htbp]
    \centering
    \includegraphics[keepaspectratio, width=\columnwidth]{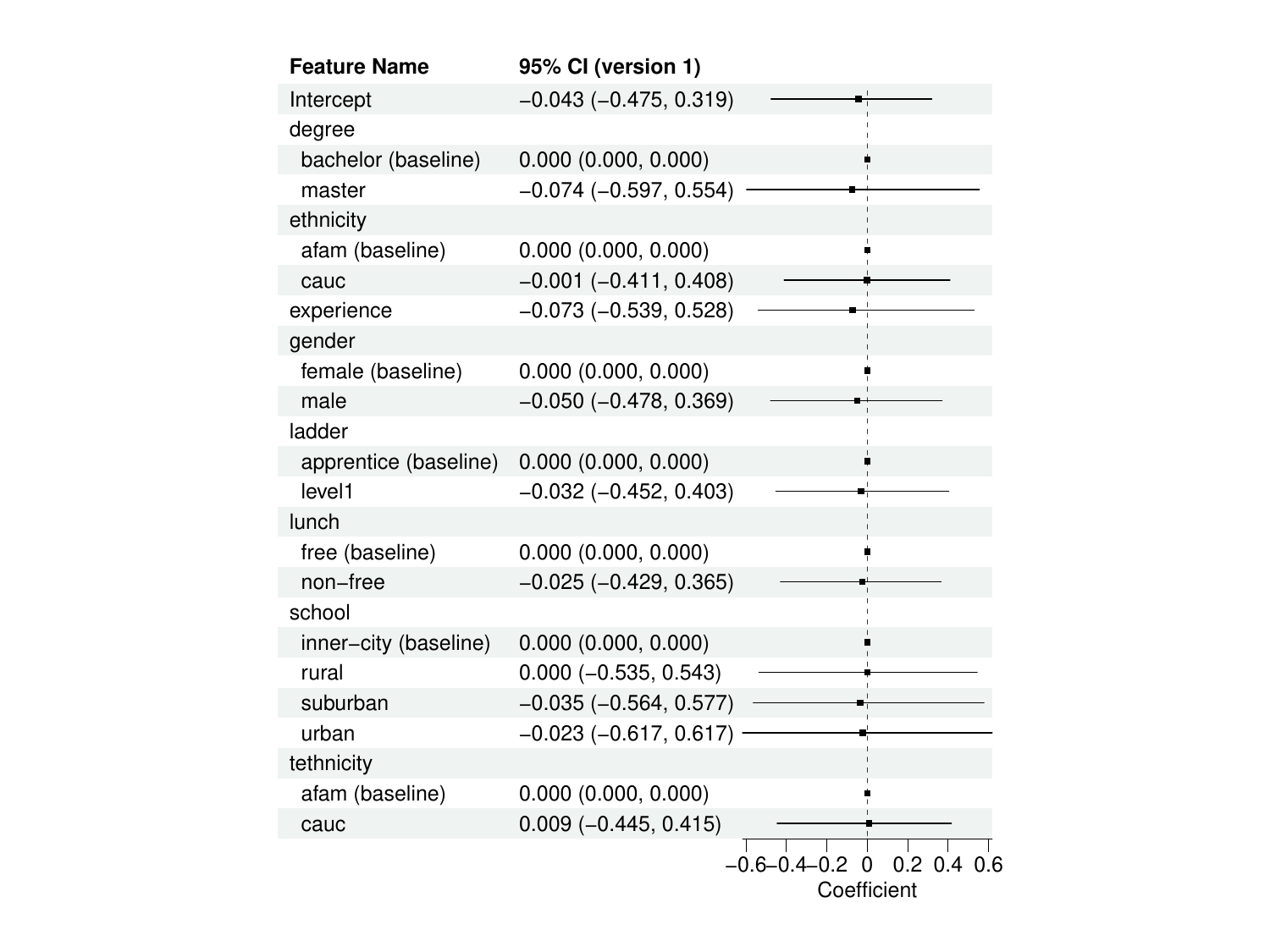}
    \caption{
    Forest plot of the estimated coefficients of the gating model for the latent version $1$ in the \texttt{small} class.
    Points denote mean estimates, and horizontal bars represent bootstrap $95\%$ confidence intervals based on 100 bootstrap resamples.
    }
    \label{fig:empirical-analysis-222}
\end{figure}

\begin{figure}[htbp]
    \centering
    \includegraphics[keepaspectratio, width=\columnwidth]{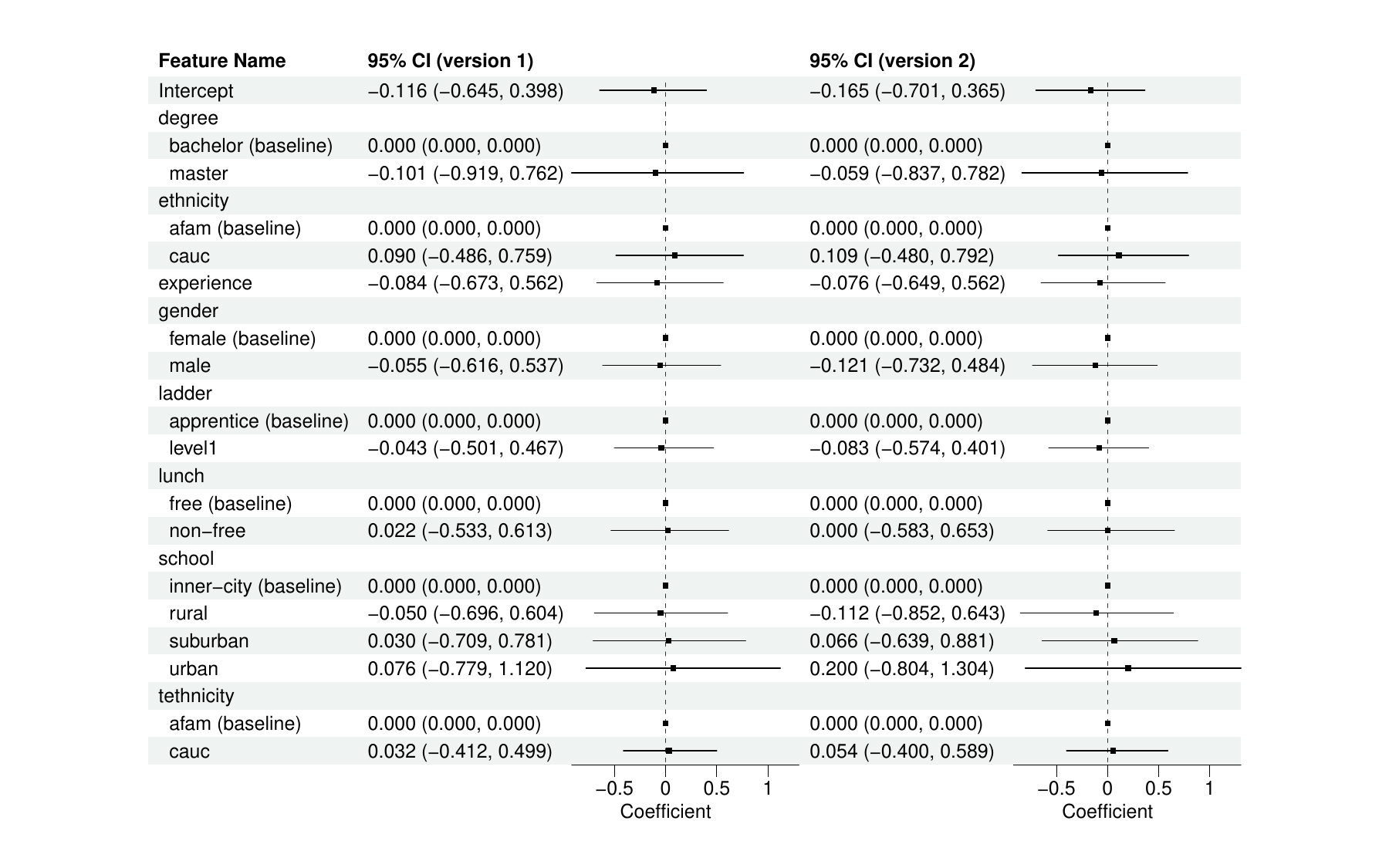}
    \caption{
    Forest plot of the estimated coefficients of the gating model for the \texttt{small} class. The plot is compared with latent versions $1$ and $2$.
    Points denote mean estimates, and horizontal bars represent bootstrap $95\%$ confidence intervals based on 100 bootstrap resamples.
    }
    \label{fig:empirical-analysis-333}
\end{figure}

%
%
\section{Discussion and Extensions}
\label{sec:discussions-and-extensions}

We have developed a new methodological framework for estimating causal effects of the latent versions by introducing the MoE framework.
This approach enables us to estimate the pairwise ACE even if the underlying versions are unobserved.
Specifically, our method statistically recovers the structure of latent versions flexibly characterized by the MoE model under Assumption \ref{asmp:causal-alignment}.
Consequently, the proposed framework allows a deeper understanding of the mechanisms of complex interventions.
This perspective provides new insights into bias reduction and stable estimation.

Our theoretical development is presented for a parametric MoE specification with Gaussian experts and linear mean structure, primarily to discuss the asymptotic properties.
Importantly, the proposed estimator and the identification do not rely on Gaussian, but they rely on the MoE representation and its identifiability.
In particular, the expert component can be extended to generalized linear models with exponential family outcomes, including the general theory in \citet{Jiang1999-vw}.
Extending our theoretical results to such broader conditions is beyond the scope of the present paper and is left for future work.

A limitation of the present study is that we assume the outcome, treatment, and latent version model are correctly specified.
An important direction for future work is to develop a doubly robust (DR) estimator by extending our method.
The usual DR estimator has the property of robustness, which remains consistent when either the outcome model or the treatment assignment model is correctly specified.
In the mixture modeling framework, however, model misspecification may affect the accuracy of posterior estimates, and its errors can propagate across models.
Therefore, the usual robustness property of DR estimators does not trivially hold.
A theoretical characterization of this mutual dependence remains an important topic for future research.

In practical applications, standard techniques such as information criteria or cross-validation may be employed to select the number of latent versions.
However, we did not conduct a theoretical analysis for selecting the number because the primary aim of this study is to propose a new estimation framework.
Establishing connections with existing work on selecting the numbers remains an important task for future exploration.

\section*{Supplementary information}

Our implementations and supplementary materials are available at \url{https://github.com/yoshikawa-kohei/latent-versions}.

\section*{Conflict of interest}

On behalf of all authors, the corresponding author states that there is no conflict of interest.

\section*{Acknowledgments}
S.K. was supported by JSPS KAKENHI grants JP23K11008, JP23H03352, JP23H00809, and JP25H01107. 

%
%
\appendix
\section{Proof of Proposition \ref{prop:propensity-score-theorem}}
\label{sec:Proof_of_Proposition1}

\begin{proof}
\begin{align*}
&\Pr\qty(
 Z_i^{(t,k)}=1
 \mid
 Y_i^{(t,k)},p_{t,k}(\bm X_i)
)\\
&\quad=
\mathbb E\qty[
 \mathbb E\qty[
  Z_i^{(t,k)}
  \mid
  Y_i^{(t,k)},p_{t,k}(\bm X_i),\bm X_i
 ]
 \middle|
 Y_i^{(t,k)},p_{t,k}(\bm X_i)
]\\
&\quad=
\mathbb E\qty[
 p_{t,k}(\bm X_i)
 \mid
 Y_i^{(t,k)},p_{t,k}(\bm X_i)
]\\
&\quad=
p_{t,k}(\bm X_i).
\end{align*}
The second equality follows from Assumption \ref{asmp:weak-conditional-exchangeability} and
$\mathbb E[Z_i^{(t,k)}\mid\bm X_i]=p_{t,k}(\bm X_i)$.
Since $Z_i^{(t,k)}$ is binary, its conditional distribution given $p_{t,k}(\bm X_i)$ does not depend on $Y_i^{(t,k)}$.
Therefore,
\begin{align*}
 Y_i^{(t,k)}
 \indep
 Z_i^{(t,k)}
 \mid
 p_{t,k}(\bm X_i).
\end{align*}
\end{proof}

\section{Proof of Proposition \ref{prop:identification}}
\label{sec:Proof_of_Proposition2}

\begin{proof}
By the law of iterated expectations, the equation
\[
\mathbb{E}\qty[Y_i^{(t,k)}]
=
\mathbb{E}_{\bm{X}_i}
\qty[
  \mathbb{E}\qty[Y_i^{(t,k)} \mid \bm{X}_i ]
]
\]
holds.
By applying Assumptions \ref{asmp:consistency}--\ref{asmp:assumption-Positivity} (consistency, weak conditional exchangeability, and positivity), the above expectation can be rewritten step by step as follows:
\begin{align*}
\psi_{t,k}
&=\mathbb{E}_{\bm{X}_i}\qty[
  \mathbb{E}\qty[Y_i^{(t,k)} \mid \bm{X}_i]
  \cdot
  \mathbb{E}\qty[
    \frac{
      Z_i^{(t,k)}
    }{
      e_t(\bm{X}_i) \rho_{t,k}(\bm{X}_i)
    }
    \mid \bm{X}_i
  ]
] && \\
&=\mathbb{E}_{\bm{X}_i}\qty[
  \mathbb{E}\qty[
    \frac{
      Y_i^{(t,k)} Z_i^{(t,k)}
    }{
      e_t(\bm{X}_i) \rho_{t,k}(\bm{X}_i)
    }
    \mid \bm{X}_i
  ]
] \ (\because \text{Assumption \ref{asmp:weak-conditional-exchangeability}}) \\
&=\mathbb{E}_{\bm{X}_i}\qty[
  \mathbb{E}\qty[
    \frac{
      Y_i Z_i^{(t,k)}
    }{
      e_t(\bm{X}_i) \rho_{t,k}(\bm{X}_i)
    }
    \mid \bm{X}_i
  ]
] \ (\because \text{Assumption \ref{asmp:consistency}})\\
&=
\mathbb{E}\qty[
  \frac{
    Z_i^{(t,k)}
  }{
    e_t(\bm{X}_i) \rho_{t,k}(\bm{X}_i)
  }Y_i
].
\end{align*}
Hence, the mean potential outcome for the treatment-version pair $(t,k)$ is identified as an expectation of the observed outcome weighted by the inverse probability.
\end{proof}

\section{Proof of Proposition \ref{prop:component-exchangeability}}
\label{sec:Proof_of_component-exchangeability}

\begin{proof}
Fix $t\in\mathcal T$ and $v\in\mathcal V^{(t)}$.
By the outcome-equivalence condition and \eqref{eq:component-potential-outcome}, $Y_i^{(t,v)}=Y_i^{(t,k)}$ almost surely for every $k\in\mathcal K_{t,v}$.
Assumption \ref{asmp:weak-conditional-exchangeability} therefore implies
\begin{align*}
  Y_i^{(t,v)}
  \indep Z_i^{(t,k)}
  \mid\bm X_i
  \quad\text{for every }k\in\mathcal K_{t,v}.
\end{align*}

By the version-alignment condition \eqref{eq:alignment-map},
\begin{align*}
  D_i^{(t)}C_i^{(t,v)}
  =
  \sum_{k\in\mathcal K_{t,v}}Z_i^{(t,k)}
  \quad\text{almost surely}.
\end{align*}
Since $D_i^{(t)}C_i^{(t,v)}$ is binary, we obtain
\begin{align*}
  &\Pr\qty(
    D_i^{(t)}C_i^{(t,v)}=1
    \mid Y_i^{(t,v)},\bm X_i
  )\\
  &\quad=
  \mathbb E\qty[
    \sum_{k\in\mathcal K_{t,v}}Z_i^{(t,k)}
    \mid Y_i^{(t,v)},\bm X_i
  ]\\
  &\quad=
  \sum_{k\in\mathcal K_{t,v}}
  \mathbb E\qty[Z_i^{(t,k)}\mid\bm X_i]\\
  &\quad=
  \Pr\qty(
    D_i^{(t)}C_i^{(t,v)}=1
    \mid\bm X_i
  )
  \quad\text{almost surely},
\end{align*}
where the second equality follows from linearity of conditional expectation and conditional independence.
Hence,
\begin{align*}
  Y_i^{(t,v)}
  \indep D_i^{(t)}C_i^{(t,v)}
  \mid\bm X_i.
\end{align*}
\end{proof}

\section{Proof of Theorem \ref{thm:consistency}}
\label{sec:Proof_of_Theorem1}

\begin{proof}
    We first prove the following lemma.

    \begin{lemma}
        Suppose that the MoE outcome model $f_{t,v}(Y_i \mid \bm{X}_i; \bm{\gamma}_{t,v})$ and the gating model  $\pi_{t,v}(\bm{X}_i; \bm{\eta}_{t,v})$ are correctly specified and identifiable under the initialization and ordering restrictions in Theorem 1 and Lemma 3(b) of \citet{Jiang1999-vw}.
        Further assume that the treatment assignment model $e_t(\bm{X}_i; \bm{\zeta})$ is correctly specified.
        Then, under standard regularity conditions,
        \begin{align*}
            (\widehat{\bm{\zeta}}, \widehat{\bm{\theta}_t}) \xrightarrow{p} (\bm{\zeta}^{*}, \bm{\theta}_t^{*}),
        \end{align*}
        where $\bm{\theta}_t^{*}$ denotes the true parameter values of the outcome and version assignment model, and $\bm{\zeta}^{*}$ denotes the true parameter values of the treatment assignment model.
        \label{lemma:consistency_model_parameters}
    \end{lemma}
    \begin{proof}
        By Theorem 1 and Lemma 3(b) of \citet{Jiang1999-vw}, the MoE is identifiable under the assumed order and initialization constraints.
        Under standard regularity conditions, the maximum likelihood estimator $\widehat{\bm{\theta}}_t$ obtained by the EM algorithm is consistent:
        \begin{align*}
            \widehat{\bm{\theta}}_t \xrightarrow{p} \bm{\theta}_t^{*}.
        \end{align*}
        Similarly, since the treatment assignment model $e_t(\bm{X}_i;\bm{\zeta})$ is correctly specified as a multinomial logistic regression, the MLE $\widehat{\bm{\zeta}}$ is also consistent:
        \begin{align*}
            \widehat{\bm{\zeta}} \xrightarrow{p} \bm{\zeta}^{*}.
        \end{align*}
        Since both estimators converge in probability to their true parameter values, the joint convergence follows directly from Slutsky's theorem.
    \end{proof}

    Next, we obtain
    \begin{align*}
      &\qty| \frac{1}{n}\sum_{i=1}^n g_i(\widehat{\bm{\zeta}}, \widehat{\bm{\theta}}_t) - \mathbb{E}\qty[g_i(\bm{\zeta}^*, \bm{\theta}_t^*)] |\\
      &= \qty| \frac{1}{n}\sum_{i=1}^n g_i(\widehat{\bm{\zeta}}, \widehat{\bm{\theta}}_t) -\mathbb{E}\qty[g_i(\widehat{\bm{\zeta}}, \widehat{\bm{\theta}}_t)] +  \mathbb{E}\qty[g_i(\widehat{\bm{\zeta}}, \widehat{\bm{\theta}}_t)] - \mathbb{E}\qty[g_i(\bm{\zeta}^*, \bm{\theta}_t^*)]|\\
      &\leq \qty| \frac{1}{n}\sum_{i=1}^n g_i(\widehat{\bm{\zeta}}, \widehat{\bm{\theta}}_t) -\mathbb{E}\qty[g_i(\widehat{\bm{\zeta}}, \widehat{\bm{\theta}}_t)] |  + \qty| \mathbb{E}\qty[g_i(\widehat{\bm{\zeta}}, \widehat{\bm{\theta}}_t)] - \mathbb{E}\qty[g_i(\bm{\zeta}^*, \bm{\theta}_t^*)]|\\
      &\leq \underbrace{\sup_{\bm{\zeta}, \bm{\theta}_t \in \Theta} \qty| \frac{1}{n}\sum_{i=1}^n g_i(\bm{\zeta}, \bm{\theta}_t) -\mathbb{E}\qty[g_i(\bm{\zeta}, \bm{\theta}_t)] |}_{\text{(i)}} + \underbrace{\qty| \mathbb{E}\qty[g_i(\widehat{\bm{\zeta}}, \widehat{\bm{\theta}}_t)] - \mathbb{E}\qty[g_i(\bm{\zeta}^*, \bm{\theta}_t^*)]|}_{\text{(ii)}}.
    \end{align*}
    
    We will show that both (i) and (ii) converge to zero in probability.
    
    The convergence (i) is shown by the uniform law of large numbers (ULLN).
    From the assumptions, 
    \begin{itemize}
      \item The sample $\qty{(Y_i, \bm{X}_i, D_i^{(t)})}_{i=1}^n$ consists of i.i.d. observations.
      \item $\Theta$ is compact.
      \item $g_i(\bm{\zeta},\bm{\theta}_t)$ is continuous at each   $(\bm{\zeta},\bm{\theta}_t) \in \Theta$ with probability one.
      \item There exists a random variable $M_i$ such that
      \begin{align*}
        \sup_{(\bm{\zeta},\bm{\theta}_t)\in\Theta}
        \qty|g_i(\bm{\zeta},\bm{\theta}_t)|
        \le M_i,
        \qquad
        \mathbb{E}\qty[M_i]<\infty.
      \end{align*} 
    \end{itemize}
    Hence, by applying the ULLN \citep[Lemma 2.4]{Newey1994-mn}, we obtain
    \begin{align*}
      \sup_{\bm{\zeta}, \bm{\theta}_t \in \Theta} \qty| \frac{1}{n}\sum_{i=1}^n g_i(\bm{\zeta}, \bm{\theta}_t) -\mathbb{E}\qty[g_i(\bm{\zeta}, \bm{\theta}_t)] | \xrightarrow{p} 0,
    \end{align*}
    and $\mathbb{E}\qty[g_i(\bm{\zeta},\bm{\theta}_t)]$ is continuous at each $(\bm{\zeta},\bm{\theta}_t) \in \Theta$.

    The convergence (ii) is derived from the continuous mapping theorem.
    Since $(\widehat{\bm{\zeta}}, \widehat{\bm{\theta}}_t)\xrightarrow{p}(\bm{\zeta}^*,\bm{\theta}_t^*)$ from Lemma \ref{lemma:consistency_model_parameters}, the continuous mapping theorem yields
    \begin{align*}
        \mathbb{E}\qty[g_i(\widehat{\bm{\zeta}},\widehat{\bm{\theta}}_t)]
      \xrightarrow{p}
        \mathbb{E}\qty[g_i(\bm{\zeta}^*,\bm{\theta}_t^*)].
    \end{align*}
    Thus, we obtain the convergence (ii):
    \begin{align*}
      \qty|
        \mathbb{E}\qty[g_i(\widehat{\bm{\zeta}}, \widehat{\bm{\theta}}_t) ]
        - \mathbb{E}\qty[g_i(\bm{\zeta}^*,\bm{\theta}_t^*)]
      |
      \xrightarrow{p} 0.
    \end{align*}
    Combining (i) and (ii) gives
    \begin{align*}
      \frac{1}{n}\sum_{i=1}^n g_i(\widehat{\bm{\zeta}}, \widehat{\bm{\theta}}_t) \xrightarrow{p} \mathbb{E}\qty[g_i(\bm{\zeta}^*, \bm{\theta}_t^*)] = \psi_{t,v}.
    \end{align*}
    Therefore, the plug-in Horvitz-Thompson estimator
    $\psi_{t,v}^{\text{HT}}(\widehat{\bm{\zeta}}, \widehat{\bm{\theta}}_t)$
    is consistent.  
\end{proof}

\section{Closed-form solutions for Gaussian outcome experts}
\label{sec:closed-form-solutions-for-gaussian-outcome-experts}
The problems \eqref{eq:mstep-gating} and \eqref{eq:mstep-expert} can be viewed as the optimization problem for the multinomial logistic regression for soft labels and a weighted Gaussian linear regression model, respectively.
In particular, we update $\bm{\gamma}_{t,v}$ for each $t \in \mathcal{T}$ and $v \in \mathcal V^{(t)}$ by solving the following problem
\begin{align*}
\bm{\gamma}_{t,v}^{\text{new}}
&=
\arg\max_{\bm{\gamma}_{t,v}}
\sum_{i=1: T_i = t}^n  r_{t,v}(Y_i, \bm{X}_i; \bm{\theta}_t)  \log f_{t,v}(Y_i\mid \bm{X}_i;\bm{\gamma}_{t,v}).
\end{align*}
For the Gaussian expert, the log-likelihood is given by
\begin{align*}
  \log f_{t,v}(Y_i\mid \bm{X}_i;\bm{\gamma}_{t,v})
  =
  -\frac{1}{2} \log(2 \pi \sigma_{t,v}^2)
  -\frac{1}{2\sigma_{t,v}^2}
      \qty(
        Y_i - \tilde{\bm{X}}_i^\top \bm{\beta}_{t,v}
        )^2.
\end{align*}
Hence the optimization problem for $\bm\gamma_{t,v}$ can be rewritten as
\begin{align*}
  \bm{\gamma}_{t,v}^{\text{new}}
  =\argmin_{\bm{\gamma}_{t,v}}
  \sum_{i: T_i=t} r_{t,v}(Y_i, \bm{X}_i; \bm{\theta}_t)
  \qty[
    \log (2\pi \sigma_{t,v}^2)
    +\frac{1}{\sigma_{t,v}^2} \qty(
      Y_i - \bm{X}_i^\top \bm{\beta}_{t,v}
    )^2
  ].
\end{align*}
By differentiating this minimization problem with respect to each parameter and setting the derivatives are equal to zero, we obtain closed-form solutions.
The explicit derivations are provided in the next subsections.

\subsection{Update for the mean parameter}

To derive the solution of $\bm{\beta}_{t,v}$, we calculate the derivative of the following objective function
\begin{align*}
    L_{\bm{\beta}_{t,v}}(\bm{\beta}_{t,v})
    =
    \sum_{i : T_i = t}
        r_{t,v}(Y_i, \bm{X}_i; \bm{\theta}_t)
        \qty[Y_i - \tilde{\bm{X}}_i^\top \bm{\beta}_{t,v}]^2.
\end{align*}
The derivative with respect to $\bm{\beta}_{t,v}$ can be calculated as
\begin{align*}
    \frac{\partial L_{\bm{\beta}_{t,v}}}{\partial \bm{\beta}_{t,v}}
    &= -2 \sum_{i : T_i = t}
        r_{t,v}(Y_i, \bm{X}_i; \bm{\theta}_t) \tilde{\bm{X}}_i
        \qty(Y_i - \tilde{\bm{X}}_i^\top \bm{\beta}_{t,v}).
\end{align*}
By setting the derivative is equal to zero and expanding equations, we obtain the following normal equation:
\begin{align*}
    \sum_{i : T_i = t} r_{t,v}(Y_i, \bm{X}_i; \bm{\theta}_t) \tilde{\bm{X}}_i \tilde{\bm{X}}_i^\top \bm{\beta}_{t,v}
    =
   \sum_{i : T_i = t} r_{t,v}(Y_i, \bm{X}_i; \bm{\theta}_t) \tilde{\bm{X}}_i Y_i.
\end{align*}
To ease the notation, we define the design matrix, weight matrix, and outcome vector for treatment $t$ as
\begin{align*}
    \bm{\Phi}_t
    &=
    \begin{bmatrix}
        \tilde{\bm{X}}_{i_1}^\top \\
        \vdots \\
        \tilde{\bm{X}}_{i_{n_t}}^\top
    \end{bmatrix},
    \quad
    \bm{Y}_t
    =
    \begin{bmatrix}
        Y_{i_1} \\
        \vdots \\
        Y_{i_{n_t}}
    \end{bmatrix},\\
    \quad
    \bm{R}_{t,v}
    &=
    \mathrm{diag}\qty[
      r_{t,v}(Y_{i_1}, \bm{X}_{i_1}; \bm{\theta}_t), \dots, r_{t,v}(Y_{i_{n_t}}, \bm{X}_{i_{n_t}}; \bm{\theta}_t)],
\end{align*}
respectively, where $\{i_1,\dots,i_{n_t}\}$ denotes the set of indices $i$ such that $T_i = t$.
By using these notations, the normal equation can be represented by
\begin{align*}
    \bm{\Phi}_t^\top \bm R_{t,v} \bm{\Phi}_t \bm{\beta}_{t,v}
    =
    \bm{\Phi}_t^\top \bm R_{t,v} \bm Y_t.
\end{align*}
Therefore, the regression coefficient $\bm{\beta}_{t,v}$ is updated by the weighted least squares solution:
\begin{align*}
  \bm{\beta}_{t,v}^{\text{new}}
  =
  \qty(\bm{\Phi}_t^\top \bm{R}_{t,v} \bm{\Phi}_t)^{-1} \bm{\Phi}_t^\top \bm{R}_{t,v} \bm{Y}_t,
\end{align*}
where $\bm{\Phi}_t^\top \bm{R}_{t,v} \bm{\Phi}_t$ is assumed to be nonsingular.

\subsection{Update for the variance parameter}

Given $\bm{\beta}_{t,v}$, to obtain $\sigma_{t,v}^2$, we minimize the following objective function
\begin{align*}
    L_{\sigma_{t,v}^2}(\sigma_{t,v}^2)
    =
    \sum_{i: T_i=t}
    r_{t,v}(Y_i, \bm{X}_i; \bm{\theta}_t)
    \qty[
      \log (2\pi \sigma_{t,v}^2)
      +\frac{1}{\sigma_{t,v}^2} \qty(
        Y_i - \tilde{\bm{X}}_i^\top \bm{\beta}_{t,v}
      )^2
    ].
\end{align*}
The derivative with respect to $\sigma_{t,v}^2$ can be calculated as
\begin{align*}
    \frac{\partial L_{\sigma_{t,v}^2}}{\partial \sigma_{t,v}^2}
    =
    \sum_{i : T_i = t}
    r_{t,v}(Y_i, \bm{X}_i; \bm{\theta}_t)
        \left[
            \frac{1}{\sigma_{t,v}^2}
            -
            \frac{1}{(\sigma_{t,v}^2)^2}
            \qty(Y_i - \tilde{\bm{X}}_i^\top \bm{\beta}_{t,v})^2
        \right].
\end{align*}
Setting this derivative is equal to zero and multiplying both sides by $\sigma_{t,v}^4$ yields
\begin{align*}
    \sum_{i : T_i = t}
    r_{t,v}(Y_i, \bm{X}_i; \bm{\theta}_t)
        \qty[
            \sigma_{t,v}^2
            -
            \qty(Y_i - \tilde{\bm{X}}_i^\top \bm{\beta}_{t,v})^2
        ]
    = 0,
\end{align*}
that is
\begin{align*}
    \sigma_{t,v}^2
    \sum_{i : T_i = t}
    r_{t,v}(Y_i, \bm{X}_i; \bm{\theta}_t)
    =
    \sum_{i : T_i = t}
    r_{t,v}(Y_i, \bm{X}_i; \bm{\theta}_t)
        \qty(Y_i - \tilde{\bm{X}}_i^\top \bm{\beta}_{t,v})^2.
\end{align*}
Therefore, we obtain the solution as follows:
\begin{align*}
    {\sigma_{t,v}^2}^{\text{new}}
    &=
    \frac{
        \sum_{i : T_i = t}
          r_{t,v}(Y_i, \bm{X}_i; \bm{\theta}_t)
            \qty(Y_i - \tilde{\bm{X}}_i^\top \bm{\beta}_{t,v})^2
    }{
        \sum_{i : T_i = t}
        r_{t,v}(Y_i, \bm{X}_i; \bm{\theta}_t)
    }\\
    &=
    \frac{
        (\bm{Y}_t - \bm{\Phi}_t \bm{\beta}_{t,v})^\top
        \bm{R}_{t,v}
        (\bm{Y}_t - \bm{\Phi}_t \bm{\beta}_{t,v})
    }{
        \bm{1}^\top \bm{R}_{t,v} \bm{1}
    },
\end{align*}
where $\bm 1$ denotes an $n_t$-dimensional vector of ones.

%
%
\bibliographystyle{apalike}
\bibliography{references}

\end{document}